\documentclass[11pt]{article}

\title{Additive One Approximation for Minimum Degree Spanning Tree: Breaking the $O(mn)$ Time Barrier}

\date{}

\author{Sayan Bhattacharya%
\thanks{
    Email: \texttt{S.Bhattacharya@warwick.ac.uk}
} \\
University of Warwick
\and Ermiya Farokhnejad%
\thanks{
    Email: \texttt{Ermiya.Farokhnejad@warwick.ac.uk}
} \\
University of Warwick
\and Haoze Wang%
\thanks{
    Email: \texttt{2200012915@stu.pku.edu.cn}
} \\
Peking University
}

\usepackage[linesnumbered,ruled,vlined]{algorithm2e}

\usepackage{fancyhdr}
\usepackage{blindtext}
\usepackage{amsthm}
\usepackage{amsmath}
\usepackage{amssymb}
\usepackage{thmtools}
\usepackage{thm-restate}
\usepackage{anyfontsize}
\usepackage{xstring}
\usepackage{xargs}
\usepackage{makeidx}
\usepackage{cancel}
\usepackage{chngcntr}
\usepackage{caption}
\usepackage{subcaption}
\usepackage{float}
\usepackage{pgf,tikz,pgfplots}
\usepackage{latexsym}
\usepackage{verbatim}
\usepackage{bookmark}
\usepackage[toc]{appendix}
\usepackage[a4paper,left=2cm,right=2cm,top=3cm,bottom=2.4cm]{geometry}

\usepackage{cleveref}

\usepackage{listings}
\usepackage{hyperref}

\hypersetup{
	colorlinks=true,
	linkcolor=blue,
	filecolor=magenta,      
	urlcolor=cyan,
}

\usepackage{algorithmicx}
\usepackage{algpseudocode}
\usepackage[dvipsnames]{xcolor}

\usepackage{bidicontour}

\newtheorem{theorem}{Theorem}[section]
\newtheorem{lemma}[theorem]{Lemma}

\newtheorem{corollary}[theorem]{Corollary}

\newtheorem{definition}[theorem]{Definition}

\newtheorem{assumption}[theorem]{Assumption}
\newtheorem{observation}[theorem]{Observation}
\newtheorem{claim}[theorem]{Claim}

\hypersetup{ colorlinks, citecolor=blue}

\PassOptionsToPackage{unicode}{hyperref}

\usepackage{framed}
\usepackage[framemethod=tikz]{mdframed}
\usepackage{tikz-cd}

\newcommand{\free}{\text{free}}

\newcommand{\calP}[0]{\mathcal{P}}

\newcommand{\calF}[0]{\mathcal{F}}

\newcommand{\calC}[0]{\mathcal{C}}

\newcommand{\calM}[0]{\mathcal{M}}

\renewcommand{\th}[0]{\text{th}}

\newcommand{\tail}[0]{\text{Tail}}
\newcommand{\trivial}[0]{\text{tr}}

\usepackage{ifthen}
\newboolean{debug}
\setboolean{debug}{false} 

\begin{document}

\maketitle

\begin{abstract}
    We consider the ``minimum degree spanning tree'' problem.
    As input, we receive an undirected, connected graph $G=(V, E)$ with $n$ nodes and $m$ edges, and our task is to find a spanning tree $T$ of $G$ that minimizes $\max_{u \in V} \deg_T(u)$, where $\deg_T(u)$ denotes the degree of $u \in V$ in $T$.
    
    The problem is known to be NP-hard. 
    In the early 1990s, an influential work by F\"{u}rer and Raghavachari presented a local search algorithm that runs in $\tilde{O}(mn)$ time, and returns a spanning tree with maximum degree at most $\Delta^\star+1$, where $\Delta^\star$ is the optimal objective. 
    This remained the state-of-the-art runtime bound for computing an additive one approximation, until now.
     
    We break this $O(mn)$ runtime barrier dating back to three decades, by providing a deterministic algorithm that returns an additive one approximate optimal spanning tree in $\tilde{O}(mn^{3/4})$ time. This constitutes a substantive progress towards answering an open question that has been repeatedly posed in the literature [Pettie'2016, Duan and Pettie'2020, Saranurak'2024]. 

    Our algorithm is based on a novel application of the {\em blocking flow} paradigm.
    
\end{abstract}

\pagenumbering{gobble}

 \newpage
 \tableofcontents

\addtocontents{toc}{\protect\setcounter{tocdepth}{2}}

\newpage
\pagenumbering{arabic}

\section{Introduction}\label{sec:intro}

Consider the {\em textbook} optimization problem~\cite{WS11} of computing a  {\em minimum degree spanning tree} (MDST). As input, we are given an undirected, connected graph $G = (V, E)$ with $|V| = n$ nodes and $|E| = m$ edges. Our goal is to compute a spanning tree $T$ of $G$ that minimizes $\max_{u \in V} \deg_T(u)$, where $\deg_T(u)$ is the degree of a node $u$ in $T$. Let $\Delta^\star$ be the optimal objective value, i.e., the optimal spanning tree $T^\star$ has maximum degree $\Delta^\star$. The problem is clearly NP-hard; even deciding whether $\Delta^\star = 2$ is equivalent to detecting whether $G$ contains a Hamiltonian path. 

In the early 1990s, F\"{u}rer and Raghavachari \cite{FR92} designed an elegant local search algorithm for this problem, which runs in $\tilde{O}(mn)$ time\footnote{Throughout the paper, we use the $\tilde{O}(.)$ notation to hide polylogarithmic in $n$ factors.} and gives an {\em additive} one approximation, i.e., it returns a spanning tree with maximum degree at most $\Delta^\star+1$. After three decades, this $\tilde{O}(mn)$ runtime bound remains the state-of-the-art. In fact, even if we allow for {\em purely multiplicative} $O(1)$-approximation, no algorithm is known to beat the runtime of \cite{FR92} in sparse graphs. Whether we can solve this problem in near-linear time is a major open question, which has repeatedly been asked in the graph algorithms literature over the years~\cite{DuanP20,P16,S24}. 

We make substantive progress towards this open question by obtaining the following result.

\begin{theorem}\label{theorem:main}
    There exists a deterministic algorithm that, given an input graph $G = (V, E)$, returns an additive one approximate minimum degree spanning tree of $G$ in $\tilde{O}\left(m n^{3/4}\right)$ time.
\end{theorem}

\paragraph{Remark.} 
Consider the following generalization of the MDST problem, which is known as {\em bounded degree spanning tree}  (BDST). As part of the input, we get an integer $b(u) \geq 1$ for each $u \in V$.
We have to either return a spanning tree $T$ of $G$ with $\deg_T(u) \leq b(u)+1$ for all  $u \in V$, or  certify that there does {\em not} exist any spanning tree $T$ of $G$ with $\deg_T(u) \leq b(u)$ for all  $u \in V$.
\Cref{theorem:main} seamlessly extends to this more general setting, without incurring any overhead in total runtime. For simplicity of exposition, however, we assume that every node $u \in V$ has the same $b(u) \equiv \Delta^\star$ value.
A summary of the generalization is provided in \Cref{appendix:generalization}.

\paragraph{Other Related Work.} 
Duan, He, and Zhang \cite{DHZ20} designed an algorithm that, in  $\tilde{O}(m/\epsilon^7)$ time, returns a spanning tree with maximum degree at most $(1+\epsilon) \Delta^\star + O(\epsilon^{-2}\log n)$. 
 In addition,
Chekuri, Quanrud, and Torres \cite{CQT21} showed how to compute a spanning tree with maximum degree at most $\lceil (1+\epsilon) \Delta^\star \rceil + 2$ in $\tilde{O}(n^2/\epsilon^2)$ time.
The techniques in these two papers, however, seem inherently incapable of breaking the $O(mn)$ time barrier for additive one approximation. We explain this in more detail in \Cref{sec:technical-overview-comparison}.
By combining our result with \cite{CQT21}, we immediately get the following corollary.\footnote{See the discussion at the end of page 20 and beginning of page 21 in the arXiv version of \cite{CQT21}, to see how the sparsification technique in \cite{CQT21} can be combined with our result to get \Cref{cor:n-dependent}.}

\begin{corollary}\label{cor:n-dependent}
    There exists an algorithm that, given an input graph $G = (V, E)$, with high probability, computes a spanning tree of $G$ in $\tilde{O}\left((m +  n^{7/4}) \cdot \epsilon^{-2}\right)$ time, whose maximum degree is bounded by $\lceil (1+\epsilon) \Delta^\star \rceil + 2$. 
\end{corollary}

Following the work of \cite{FR92}, a sequence of influential papers studied a more general version of the problem~\cite{RMRRH93,KR00,CRRT05,Goemans06,SL07}: Here, the input graph is (edge)-weighted, and we want to compute a spanning tree of minimum total cost, subject to an upper bound on its maximum degree. This line of work culminated with the algorithm of Singh and Lau~\cite{SL07}, who showed how to obtain a  spanning tree with maximum degree $\Delta^\star+1$, whose cost is upper bounded by the minimum possible cost of any spanning tree with maximum degree $\Delta^\star$, in polynomial time. This result by \cite{SL07} is among the most celebrated applications of the {\em iterated rounding} technique~\cite{LauRS11}.

\section{Notations and Preliminaries}\label{preliminaries}

We start by defining some notations that will be used throughout the rest of this paper. Let  $G=(V, E)$ denote an undirected and {\em connected} input graph, with $n := |V|$ nodes and $m := |E|$ edges.
Consider any subgraph $H$ of $G$. We let $V(H) \subseteq V$ and $E(H) \subseteq E$ respectively denote the sets of nodes and edges of $H$, and we write $H \subseteq G$.
For every node $v \in V(H)$, we let $\deg_H(v)$ and $\psi_H(v)$ respectively denote the degree and the set of neighbors of the node $v$ in $H$.
Whenever we use the phrase ``{\bf component} of $H$'', we refer to a subgraph ({\em not} a node-set) that happens to be a connected component of $H$. Whenever we use the term ``{\bf forest}'', we refer to a forest $\calF \subseteq G$ defined on the {\em entire} node-set $V$ (i.e., $V(\calF) = V$). In contrast, whenever we use the term ``{\bf sub-tree}'' of $H \subseteq G$, we refer to a tree $T \subseteq H$ defined over a {\em subset} of nodes in $V(H)$ (i.e., $V(T) \subseteq V(H)$).
Finally, given any subset $S \subseteq V$, we let $G[S]$ denote the subgraph of $G$ induced by $S$. Thus, we have $V(G[S]) = S$ and $E(G[S]) = \{ (u, v) \in E : u, v \in S\}$.

In the ``{\bf minimum degree spanning tree}'' problem, our goal is to compute a spanning tree $T$ of the input graph $G$ which minimizes $\max_{u \in V} \{ \deg_T(u) \}$. Let $\Delta^\star$ denote the optimal objective value, and $T^\star$ denote the optimal spanning tree. Thus, we have 
$$ T^\star = \arg\min_{\substack{T \subseteq G: \\ T \text{ is a spanning tree of } G} } \left\{ \max_{u \in V} \deg_{T}(u) \right\}, \text{ and } \Delta^* = \max_{u \in V} \left \{ \deg_{T^\star}(u) \right \}.$$

We say that a forest $\calF$ is {\bf valid} iff $\deg_{\calF}(u) \leq \Delta^\star+1$ for all nodes $u 
\in V$. We will show how to compute a valid forest that is also a spanning tree of $G$ in $\tilde{O}(mn^{3/4})$ time (see \Cref{theorem:main}). 

\paragraph{Knowledge of $\Delta^\star$.} For ease of exposition, throughout the rest of the paper, we assume that we {\em know} the value of $\Delta^\star$, and our goal is to {\em find} a spanning tree of $G$ with maximum degree at most $\Delta^\star+1$. This assumption is w.l.o.g.: Once we obtain an algorithm which requires the knowledge of $\Delta^\star$, we can convert it into an algorithm that does not require this knowledge by doing a simple binary search on $\Delta^\star$. This incurs only a $O(\log n)$ factor overhead in the total running time, since $\Delta^\star \leq n$.
In \Cref{appendix:unknown-Delta-star}, we provide a more detailed explanation of how this binary search is performed.

\subsection{Basic Building Blocks: Molecular Decomposition and Atoms}\label{building-blocks}
\label{sec:molecules-atoms}

Our algorithm relies on a new forest decomposition technique which we call a ``molecular decomposition''. 
The decomposition consists of some mutually node-disjoint sub-trees, which we refer to as ``molecules''. Each molecule, in turn, further contains some mutually node-disjoint sub-trees that we refer to as ``atoms''. In this section, we precisely define these relevant terminologies, and illustrate a key insight from the \cite{FR92} algorithm along the way.

Fix  any given {\em valid} forest $\calF$. 
Consider any two distinct nodes $u$ and $v$ within the same component of $\calF$.
We denote by $P_{u,v}^\calF$ the unique path between $u$ and $v$ in $\calF$.
Let $T_{u \leftarrow v}^\calF$ be the connected sub-tree of $\calF$ containing $u$, obtained by removing the unique edge incident on $v$ in  $P_{u,v}^\calF$.
Note that  $T_{u \leftarrow v}^\calF = T_{x \leftarrow v}^\calF$, where $(x, v)$ is the unique edge incident on $v$ in $P_{u,v}^\calF$.
We denote $T_{x \leftarrow v}^\calF$ by $T_{xv}^\calF$, to indicate that there is an edge  $(x, v) \in E(\calF)$. The following illustrates an example of these sub-trees.
\Cref{fig:Tuv} illustrates an example of these sub-trees.

For every  edge $(x,y) \in E(\calF)$, we say that $T_{xy}^\calF$ is a {\bf normal molecule} of $\calF$, with $y$ being its {\bf root}. Thus, a normal molecule is a sub-tree of $\calF$ that is connected to the rest of $\calF$ via a single edge.
For every component $C$ of $\calF$, we say that $C$ is a {\bf special molecule} of $\calF$.
We use the term {\bf molecule} to refer to an entity that is either a normal molecule or a special molecule, w.r.t.~an underlying forest in $G$.\footnote{Note that a normal molecule has a root, and a special molecule does {\em not} have a root.} 
We refer to a collection $\calM$ of molecules of $\calF$ as a {\bf molecular decomposition} if it satisfies the following conditions;
\begin{itemize}
    \item Molecules in $\calM$ are mutually node-disjoint.

    \item The root $y$ of every normal molecule $T^\calF_{xy}$ in $\calM$ is not contained in any other molecule in $\calM$.
\end{itemize}
The molecules in $\calM$ are called {\bf $\calM$-molecules}. 
A node $u \in V$ is {\bf $\calM$-free} if it is {\em not} contained in any $\calM$-molecule, and {\bf $\calM$-covered} otherwise. We further classify each $\calM$-covered node into one of two categories -- {\bf $\calM$-reducible} and {\bf $\calM$-non-reducible} -- as described below.

\begin{figure}[ht!]
\caption{An example of a sub-tree $T^\calF_{u\leftarrow v}$.}
    \label{fig:Tuv}
\centering

\tikzset{every picture/.style={line width=0.75pt}} 

\begin{tikzpicture}[x=0.75pt,y=0.75pt,yscale=-0.8,xscale=0.8]

\draw [line width=0.75]    (359,108) -- (328,146) ;
\draw [line width=0.75]    (328,146) -- (298.33,134.67) ;
\draw [line width=0.75]    (328,146) -- (303.33,179.67) ;
\draw  [fill={rgb, 255:red, 0; green, 0; blue, 0 }  ,fill opacity=1 ] (325.08,146) .. controls (325.08,144.39) and (326.39,143.08) .. (328,143.08) .. controls (329.61,143.08) and (330.92,144.39) .. (330.92,146) .. controls (330.92,147.61) and (329.61,148.92) .. (328,148.92) .. controls (326.39,148.92) and (325.08,147.61) .. (325.08,146) -- cycle ;
\draw  [fill={rgb, 255:red, 0; green, 0; blue, 0 }  ,fill opacity=1 ] (295.42,134.67) .. controls (295.42,133.06) and (296.72,131.75) .. (298.33,131.75) .. controls (299.94,131.75) and (301.25,133.06) .. (301.25,134.67) .. controls (301.25,136.28) and (299.94,137.58) .. (298.33,137.58) .. controls (296.72,137.58) and (295.42,136.28) .. (295.42,134.67) -- cycle ;
\draw  [fill={rgb, 255:red, 0; green, 0; blue, 0 }  ,fill opacity=1 ] (300.42,179.67) .. controls (300.42,178.06) and (301.72,176.75) .. (303.33,176.75) .. controls (304.94,176.75) and (306.25,178.06) .. (306.25,179.67) .. controls (306.25,181.28) and (304.94,182.58) .. (303.33,182.58) .. controls (301.72,182.58) and (300.42,181.28) .. (300.42,179.67) -- cycle ;
\draw [color={rgb, 255:red, 0; green, 0; blue, 0 }  ,draw opacity=1 ][line width=0.75]    (392.33,139.67) -- (359,108) ;
\draw [line width=0.75]    (438.33,126.67) -- (392.33,139.67) ;
\draw [line width=0.75]    (359,108) -- (324.33,85.67) ;
\draw  [fill={rgb, 255:red, 0; green, 0; blue, 0 }  ,fill opacity=1 ] (321.42,85.67) .. controls (321.42,84.06) and (322.72,82.75) .. (324.33,82.75) .. controls (325.94,82.75) and (327.25,84.06) .. (327.25,85.67) .. controls (327.25,87.28) and (325.94,88.58) .. (324.33,88.58) .. controls (322.72,88.58) and (321.42,87.28) .. (321.42,85.67) -- cycle ;
\draw  [fill={rgb, 255:red, 0; green, 0; blue, 0 }  ,fill opacity=1 ] (389.42,139.67) .. controls (389.42,138.06) and (390.72,136.75) .. (392.33,136.75) .. controls (393.94,136.75) and (395.25,138.06) .. (395.25,139.67) .. controls (395.25,141.28) and (393.94,142.58) .. (392.33,142.58) .. controls (390.72,142.58) and (389.42,141.28) .. (389.42,139.67) -- cycle ;
\draw  [fill={rgb, 255:red, 0; green, 0; blue, 0 }  ,fill opacity=1 ] (435.42,126.67) .. controls (435.42,125.06) and (436.72,123.75) .. (438.33,123.75) .. controls (439.94,123.75) and (441.25,125.06) .. (441.25,126.67) .. controls (441.25,128.28) and (439.94,129.58) .. (438.33,129.58) .. controls (436.72,129.58) and (435.42,128.28) .. (435.42,126.67) -- cycle ;
\draw [line width=0.75]    (392.33,139.67) -- (416.33,178.67) ;
\draw  [fill={rgb, 255:red, 0; green, 0; blue, 0 }  ,fill opacity=1 ] (413.42,178.67) .. controls (413.42,177.06) and (414.72,175.75) .. (416.33,175.75) .. controls (417.94,175.75) and (419.25,177.06) .. (419.25,178.67) .. controls (419.25,180.28) and (417.94,181.58) .. (416.33,181.58) .. controls (414.72,181.58) and (413.42,180.28) .. (413.42,178.67) -- cycle ;
\draw  [color={rgb, 255:red, 0; green, 0; blue, 255 }  ,draw opacity=1 ] (280.11,202.82) .. controls (255.13,187.45) and (254.05,143.85) .. (277.7,105.45) .. controls (301.34,67.05) and (340.75,48.38) .. (365.73,63.76) .. controls (390.71,79.14) and (391.79,122.73) .. (368.14,161.13) .. controls (344.5,199.54) and (305.09,218.2) .. (280.11,202.82) -- cycle ;
\draw  [fill={rgb, 255:red, 0; green, 0; blue, 0 }  ,fill opacity=1 ] (356.08,108) .. controls (356.08,106.39) and (357.39,105.08) .. (359,105.08) .. controls (360.61,105.08) and (361.92,106.39) .. (361.92,108) .. controls (361.92,109.61) and (360.61,110.92) .. (359,110.92) .. controls (357.39,110.92) and (356.08,109.61) .. (356.08,108) -- cycle ;

\draw (209,216.4) node [anchor=north west][inner sep=0.75pt]    {$T^\calF_{u\leftarrow v} =T^\calF_{u'\leftarrow v} =T^\calF_{x\leftarrow v} =T^\calF_{xv}$};
\draw (335,145.4) node [anchor=north west][inner sep=0.75pt]    {$u$};
\draw (282,168.4) node [anchor=north west][inner sep=0.75pt]    {$u'$};
\draw (390,122.4) node [anchor=north west][inner sep=0.75pt]    {$v$};
\draw (360,87.4) node [anchor=north west][inner sep=0.75pt]    {$x$};

\end{tikzpicture}
    
\end{figure}
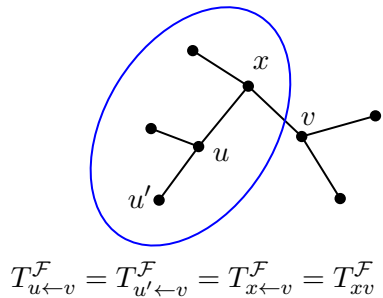

Let $T$ be any arbitrary $\calM$-molecule.
Consider the following procedure, which is essentially the algorithm of \cite{FR92} running on $T$.
Initially, we mark all of the nodes $u \in V(T)$ satisfying $\deg_{\calF}(u) \leq \Delta^\star$ as a singleton \textit{atom}, and the rest of the nodes in $T$ as \textit{bad nodes}.
Then, as long as there exists an edge $(x,y) \in E$\footnote{This can be either a forest edge $(x,y) \in E(T)$ or a non-forest edge $(x,y) \in E - E(T)$.} between two nodes $x$ and $y$ that are contained in two different atoms, we consider the path $P_{x,y}^T$ between $x$ and $y$ in $T$, and {\em merge} all of the atoms {\em hitting} $P_{x,y}^T$ (i.e.,  with at least one node on this path) together with all bad nodes on $P_{x,y}^T$, to form a new larger atom. 
Throughout this process, each atom remains a sub-tree of $\calF$, and different atoms remain mutually node-disjoint.

After running the above procedure on every $\calM$-molecule of $\calF$, we refer to every existing atom at the end as an {\bf $\calM$-atom}.  
We refer to the nodes inside these $\calM$-atoms as {\bf $\calM$-reducible}.
Observe that every $\calM$-reducible node must be $\calM$-covered.
If an $\calM$-covered node  is not $\calM$-reducible, then we say that it is {\bf $\calM$-non-reducible}.
Morally, a node $u$ is $\calM$-reducible if we can locally change $\calF$, by inserting and removing some edges inside the $\calM$-atom containing $u$, to achieve another {\em valid} forest $\calF'$ such that $\deg_{\calF'}(u) \leq \Delta^\star$.
This key property is summarized in \Cref{lem:degree-reduction}.
\Cref{fig:atoms} illustrates an example of a molecule and its atoms, and \Cref{fig:degree-reduction} illustrates the status of the molecule after applying \Cref{lem:degree-reduction} to reduce the degree of $u$.

\begin{figure}[ht]
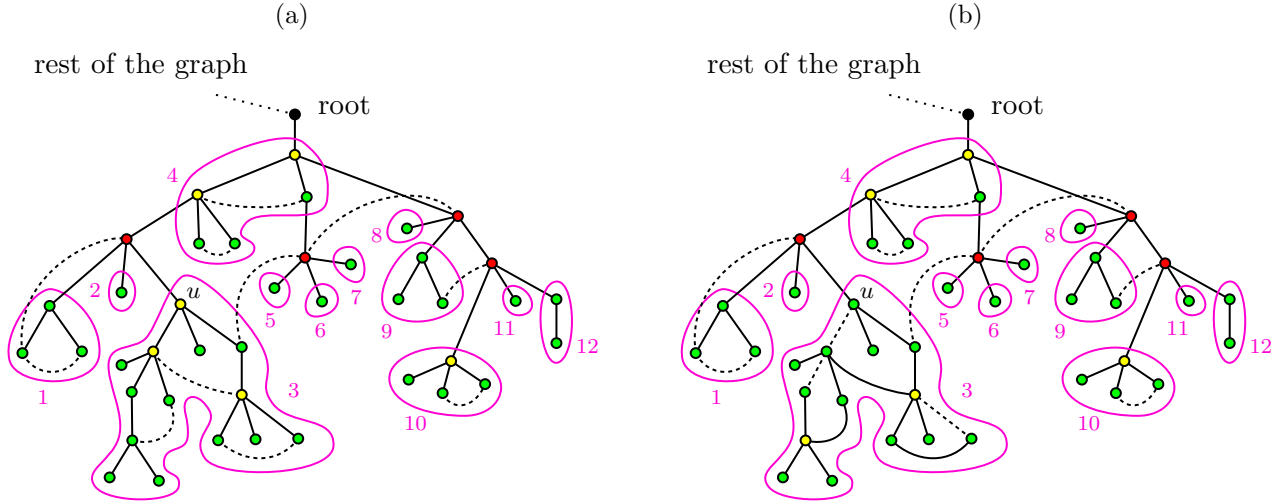

    \caption{An example of a molecule and its atoms.
    Solid line are forest edges, ans dashed lines are non-forest edges.
    The value of $\Delta^\star$ is $3$.
    Red nodes correspond to non-reducible nodes.
    Yellow nodes are reducible nodes of degree $\Delta^\star+1$.
    Green nodes are reducible nodes of degree $\leq \Delta^\star$.
    Atoms are depicted with magenta curves. (a) The status of the initial molecule. (b) The status of the molecule after reducing the degree of $u$ in atom $3$.
    The sub-tree inside atom $3$ will change but all of the other atoms and non-reducible nodes remain unaffected.}
  \centering
  \subfloat[]{\input{figures/fig-atoms}\label{fig:atoms}}
  \hfill
  \subfloat[]{\input{figures/fig-degree-reduction}\label{fig:degree-reduction}}
  \label{fig:atoms-deg-reduction}
\end{figure}

\begin{lemma}[\cite{FR92}]\label{lem:degree-reduction} There is a {\bf degree-reduction} subroutine which works as follows. Let $\calM$ be a molecular decomposition of a valid forest $\calF$, and let $C$ be any $\calM$-atom. The subroutine takes $(\calF, \calM, u)$ as input, and modifies $\calF$ by inserting/deleting some edges $e \in E$ whose {\em both} endpoints lie in $V(C)$. Let $\calF^-$ (resp.~$\calF^+$)  denote the state of $\calF$ just before the call to the subroutine (resp.~just after the subroutine finishes execution). The subroutine runs in $\tilde{O}\left(\left| E\big( G[V(C)] \big) \right| \right)$ time, and guarantees that:
    \begin{enumerate}
    \item  $\deg_{\calF^+}(u) \leq \Delta^\star$.
    \item $\calF^+$ remains a valid forest. 
    \end{enumerate}
\end{lemma}

\begin{proof}(Sketch)
    Consider the procedure described above, which defines the notion of an $\calM$-atom. 
    We claim that: At every point in time during the procedure,  for every atom $R$ and every node $z \in V(R)$, we can change the edges inside $R$ locally, in such a way that the degree of $z$ becomes  $\leq \Delta^\star$, and the degree of every other node contained in $R$ remains $\leq \Delta^\star + 1$ and changes by at most one.
    
    To see why this is true, note that in the beginning the atoms are singleton nodes with degree at most $\Delta^\star$, and so they satisfy the claim.
    When two atoms $R_x$ and $R_y$ containing the nodes $x$ and $y$ are getting merged because there is a forest edge $(x,y) \in E(T)$, the claim remains obviously correct for the new atom $R_x \cup R_y$.
    Subsequently, 
    whenever a bad node $z$ becomes part of an atom, there is a non-forest edge $(x,y) \in E - E(\calF)$ where $x$ and $y$ are in two different atoms and $z$ lies on the path $P_{x,y}^T$.
    Thus, we can reduce the degree of $x$ and $y$ in their respective atoms recursively, insert the edge $(x,y)$ into $R$, and remove an edge incident on $z$ in the path $P_{x,y}^T$ from $R$.
    Hence, the degrees of $z$ in $R$ becomes $\leq \Delta^\star$. Applying an induction hypothesis, we infer that the degree of every other node in $R$ remains $\leq \Delta^\star$ and changes by at most one.
    
    When the above procedure terminates, we set $R = C$. The runtime guarantee follows if we use standard data structures to efficiently implement this procedure.
\end{proof}

\section{Technical Overview}
\label{sec:technical:overview}

The \cite{FR92} algorithm starts with  an arbitrary spanning tree $T$ of the input graph $G$, and then in each successive iteration updates  $T$ to reduce its number of nodes with  maximum degree. At all times, it is ensured that $T$ remains a spanning tree of $G$. We, however, approach the problem from a different perspective: We start with a valid forest $\calF$ where each node is a singleton component, and $E\left(\calF \right) = \emptyset$. In each successive iteration, we update $\calF$ (by inserting/deleting some edges in it) so as to reduce the number of components of $\calF$. Throughout our algorithm, we ensure that $\calF$ remains a valid forest. We terminate when $\calF$ consists of one component, i.e., when $\calF$ becomes a spanning tree of $G$. 

The rest of this section is organized as follows. In \Cref{sec:technical-overview-new-perspective}, we present and analyze the \cite{FR92} algorithm from our new perspective. Subsequently, in \Cref{sec:technical-overview-key-insights}, we present a key technical insight under a simplifying assumption, which hints at the possibility of obtaining a ``polynomial improvement'' over the runtime bound of \cite{FR92}. Finally, in \Cref{sec:technical-overview-overcome-assumption}, we provide a very high-level sketch of several ideas that are needed to get rid of the simplifying assumption.

\subsection{The~\cite{FR92} Algorithm (From a New Perspective)}\label{sec:technical-overview-new-perspective}

The algorithm starts with a trivial valid forest $\calF$, where each node $v \in V$ is a singleton component and $E(\calF) = \emptyset$. Subsequently, it runs for $n-1$ {\bf iterations}. Each iteration invokes \Cref{lem:alg-small-scale-main} to reduce the number of components of $\calF$ by one, while ensuring that $\calF$ continues to remain a valid forest. At the end of the last iteration, there is only one component in  $\calF$, and hence $\calF$ becomes a spanning tree of $G$ with maximum degree at most $\Delta^\star+1$.  Since each iteration takes $\tilde{O}(m)$ time and there are $n-1$ iterations, the overall runtime of the algorithm is $\tilde{O}(mn)$.

\begin{lemma}[\cite{FR92}]\label{lem:alg-small-scale-main}
    There exists a deterministic algorithm that, given a valid forest $\calF$ (which is not a spanning tree), modifies $\calF$ (by inserting/deleting some edges) so that $\calF$ continues to remain a valid forest, but with one fewer component than before. The algorithm runs in $\tilde{O}(m)$ time.
\end{lemma}

\begin{proof}(Sketch)
We consider a trivial molecular decomposition $\calM_{\trivial}$ of $\calF$ where each  component of $\calF$ is a special molecule, and there is no normal molecule.
We compute all the $\calM_{\trivial}$-atoms.
Now, we say that an edge $(u, v) \in E - E(\calF)$ is {\em augmenting} iff $u$ and $v$ are contained in two different $\calM_{\trivial}$-atoms in two different components of $\calF$.
By iterating over all edges of $G = (V, E)$, we find a suitable edge $(u, v) \in E$.
Next, we apply \Cref{lem:degree-reduction} separately  on the  respective $\calM_{\trivial}$-atoms containing $u$ and $v$, so as to ensure that $\deg_{\calF}(u) \leq \Delta^\star$ and $\deg_{\calF}(v) \leq \Delta^\star$. 
Finally, we  insert the edge  into $E(\calF)$, and terminate. 
It is easy to verify that $\calF$ continues to remain a valid forest, with one fewer component than before.  
To prove the correctness of the algorithm, it now remains to show that we can always find an augmenting edge $(u, v)$ to add to $E(\calF)$.

\medskip
\noindent {\bf Correctness.}
For the sake of contradiction, suppose that there exists a valid forest  $\calF$ with $f > 1$ components, which does not admit {\em any} augmenting edge.
Let $N \subseteq V$ denote the set of all $\calM_\trivial$-non-reducible nodes of $\calF$.
Observe that if we delete all the nodes in  $N$ from $\calF$, then every component in the resulting forest would be an $\calM_\trivial$-atom. 
Now, consider the collection $\calC$, which consists of all the $\calM_\trivial$-atoms of $\calF$, plus all the nodes in $N$.
Since $\deg_\calF(u) = \Delta^\star+1$ for all $u \in N$, we infer that  
\begin{equation}
\label{eq:collection}
|\calC| \geq f + \Delta^\star \cdot |N|.
\end{equation}

Since we have assumed that there is no augmenting edge, there is no edge in $E$ which connects two distinct $\calM_{\trivial}$-atoms between two different components of $\calF$. Furthermore, by definition, there is no edge in $E$ which connects two distinct $\calM_{\trivial}$-atoms within the same component of $\calF$. Thus, it follows that {\em there is no edge in $E$ whose two endpoints belong to two distinct $\calM_{\trivial}$-atoms}.

Now, consider the optimal spanning tree $T^\star$ of $G$, which has maximum degree $\Delta^\star$. Clearly, this spanning tree $T^\star$ connects the $\calM_{\trivial}$-atoms and the nodes in $N$ together with the rest of the nodes. Thus, there exists at least $\left|\calC\right| - 1$ edges in $T^\star$ that have one endpoint in some $x \in \calC$ ($x$ is either a node in $N$ or an $\calM_{\trivial}$-atom) and the other endpoint outside of $x$. From the above discussion, we infer that each such edge is incident on $N$. Thus, from \Cref{eq:collection}, we get
\begin{equation}
\label{eq:collection1}
\Delta^\star \cdot |N| \geq |\calC| - 1 \geq f + \Delta^\star \cdot |N| - 1.
\end{equation}
Rearranging the terms in the above inequality, we get $f \leq 1$, which gives us the desired contradiction.

\medskip
\noindent
\textbf{Running Time.}
The time taken to compute all the $\calM_{\trivial}$-atoms is  $\tilde{O}(m)$.
Searching for the suitable edge $(u, v)$ takes $\tilde{O}(m)$ time, since we just need to check whether both $u$ and $v$ are in two different $\calM_\trivial$-atoms.
Finally, invoking \Cref{lem:degree-reduction} to reduce the degree of $u$ and $v$ also takes $\tilde{O}(m)$ time. 
\end{proof}

\subsection{Our Key Insight (Under a Simplifying Assumption)}\label{sec:technical-overview-key-insights}

Our quest for achieving a polynomial improvement over the $\tilde{O}(mn)$ runtime bound of \cite{FR92} starts with a natural question: Can we design an algorithmic framework which works in {\bf phases}, such that  each phase merges together a {\em polynomial} number of components of $\calF$ in $\tilde{O}(m)$ time? This would imply that the number of phases is polynomially smaller than $n$, and since each phase runs in $\tilde{O}(m)$ time, we would accordingly be able to beat the $\tilde{O}(mn)$ time barrier by a polynomial factor. In this section, we outline how to implement such a phase, under a simplifying assumption as summarized below.

\begin{assumption}\label{assumption:uniform-size}
   There are $f \geq \sqrt{n}$ components in $\calF$, with $\leq 100 \cdot n/f$ nodes in each component. 
\end{assumption}

Intuitively, \Cref{assumption:uniform-size} says that the size (in terms of the number of nodes) of every component in the forest $\calF$ is at most a constant times the average size of a component. Within this context, \Cref{lem:simple-alg-uniform-size} captures a key technical insight, and points out how to implement a phase.

\paragraph{Breaking the $O(mn)$ Time Barrier.}\label{paragraph:break-time-barrier} We start with a valid forest $\calF$ with $f = n$ components and $E(\calF) = \emptyset$. Suppose that we get extremely lucky, in the sense that \Cref{assumption:uniform-size}  continues to hold all the time. Then, as long as $f \geq n^{0.51}$, we invoke \Cref{lem:simple-alg-uniform-size} to reduce the number of components in $\calF$ by $\Omega(f^2/n) = \Omega\left(n^{0.02}\right)$, which takes $\tilde{O}(m)$ time. In other words, we spend {\em amortized} $\tilde{O}\left(m/n^{0.02}\right)$ time to reduce the number of components of $\calF$ by one. This is precisely the source of our ``polynomial advantage'' over the \cite{FR92} algorithm. Overall, the total time we spend until $f$ becomes smaller than the threshold $n^{0.51}$ is at most $\tilde{O}\left(m n/n^{0.02}\right) = \tilde{O}\left(mn^{0.98}\right)$. Subsequently, when $f$ becomes smaller than $n^{0.51}$, we switch back to the \cite{FR92} algorithm as outlined in \Cref{sec:technical-overview-new-perspective}. From this point onward, we spend $\tilde{O}(m)$ time per iteration, to reduce the number of components of $\calF$ by one. So, overall it takes $\tilde{O}\left(mn^{0.51}\right)$ time to reduce the number of components of $\calF$ all the way down to one, at which point our algorithm terminates. Clearly, the overall runtime of the entire procedure is at most $\tilde{O}\left(mn^{0.98}\right) + \tilde{O}\left(mn^{0.51}\right) = \tilde{O}\left(mn^{0.98} \right)$. 

In fact, a more refined analysis here would imply that this algorithm actually runs in $\tilde{O}(m \sqrt{n})$ time. 
We, however, skip this more refined analysis, since (i) our goal in this technical overview is only to highlight the key ideas that lead us to break the $O(mn)$ time barrier, and (ii) if we remove our simplifying assumption, then we get a slightly weaker guarantee than that of \Cref{lem:simple-alg-uniform-size}, which eventually gives us a runtime bound of $\tilde{O}\left( mn^{3/4}\right)$.

\begin{lemma}\label{lem:simple-alg-uniform-size}
    There is a deterministic algorithm that, given a valid forest $\calF$ satisfying \Cref{assumption:uniform-size} as input, modifies $\calF$ (by inserting/deleting some edges) in such a way that $\calF$ continues to remain a valid forest, but with  $\Omega(f^2/n)$ fewer components than before. The algorithm runs in $\tilde{O}(m)$ time.
\end{lemma}

We devote the rest of \Cref{sec:technical-overview-key-insights} to the proof of \Cref{lem:simple-alg-uniform-size}. We start with an important definition.

\begin{definition}
\label{def:augmenting}
Let $\calF'$ be a valid forest and $\calM'$ be a molecular decomposition of $\calF$.
We refer to a non-forest edge $(u,v) \in E - E(\calF')$ as an {\bf augmenting edge} w.r.t.~$(\calF', \calM')$ iff $u, v$ belong to two different components of $\calF$, and either
\begin{itemize}
\item  (i) $u, v$ are part of two distinct $\calM'$-atoms, 
or
\item  (ii) $u$ belongs to an $\calM'$-atom and $v$ is an $\calM'$-free node with $\deg_{\calF'}(v) \leq \Delta^\star$.
\end{itemize}
\end{definition}

Our algorithm for \Cref{lem:simple-alg-uniform-size} works in {\bf rounds}. Before the very first round, we initialize $\calM$ to be a trivial molecular decomposition of the input forest $\calF$, where every component of $\calF$ becomes a special $\calM$-molecule, and there is no normal $\calM$-molecule.\footnote{Throughout the algorithm, it will continue to hold that $\calM$ is a subset of a trivial molecular decomposition of $\calF$. In other words, every $\calM$-molecule is a component of $\calF$ (i.e., a special molecule), but {\em not} the other way round (i.e., it is {\em not} necessarily the case that every component of $\calF$ belongs to  $\calM$ as a special molecule in subsequent rounds).} Subsequently,  each round works as follows.

\paragraph{Implementing a given round.} First, we determine whether or not there exists an augmenting edge w.r.t.~$(\calF, \calM)$. If the answer is no, then we {\bf terminate} the algorithm. Otherwise, if the answer is yes, then we find an augmenting edge $(u, v) \in E - E(\calF)$, where $u$ belongs to an $\calM$-atom $C_u$ which is part of a an $\calM$-molecule $M_u$ (say). Further, the nodes $u, v$ belong to different components of $\calF$. 

Now, we fork into one of the following two cases.

\medskip
\noindent 
{\em Case (i): The node $v$ belongs to an $\calM$-atom $C_v$ which is part of an $\calM$-molecule $M_v$ (say).}  

\smallskip
\noindent In this case, we  invoke the degree-reduction subroutine from \Cref{lem:degree-reduction}  on $(\calF, \calM, u)$ and $(\calF, \calM, v)$, one after another. This modifies the forest $\calF$ while keeping it valid, so as to ensure that $\deg_{\calF}(u) \leq \Delta^\star$ and $\deg_{\calF}(v) \leq \Delta^\star$. Next, we add the edge $(u, v)$ to the forest, by setting $E(\calF) \leftarrow E(\calF) \cup \{(u, v)\}$. Clearly, $\calF$ remains a valid forest even after this step, but the number of components of $\calF$ reduces by one. We now update $\calM$ by setting $\calM \leftarrow \calM - \{ M_u, M_v\}$. At this point, we start the next round.

\medskip
\noindent 
{\em Case (ii): The node $v$ is an $\calM$-free node with $\deg_{\calF}(v) \leq \Delta^\star$.}  

\smallskip
\noindent 
In this case, we  invoke the degree-reduction subroutine from \Cref{lem:degree-reduction}  on $(\calF, \calM, u)$. This modifies the forest $\calF$ while keeping it valid, so as to ensure that $\deg_{\calF}(u) \leq \Delta^\star$. Next, we add the edge $(u, v)$ to the forest, by setting $E(\calF) \leftarrow E(\calF) \cup \{(u, v)\}$. Clearly, $\calF$ remains a valid forest even after this step, but the number of components of $\calF$ reduces by one. We now update $\calM$ by setting $\calM \leftarrow \calM - \{ M_u\}$. At this point, we start the next round.

\medskip
Clearly, each round reduces the components of $\calF$ by one, while ensuring that $\calF$ remains a valid forest. Hence, \Cref{lem:simple-alg-uniform-size}  follows from  \Cref{lm:rounds:sketch} and \Cref{lm:runtime:sketch} (the reader might find it instructive to compare their proofs against that of \Cref{lem:alg-small-scale-main}). The proof of \Cref{lm:rounds:sketch} appears in \Cref{sec:lm:rounds:sketch}.

\begin{lemma}
\label{lm:rounds:sketch}
The above algorithm runs for at least $f^2/n$ rounds.
\end{lemma}

\begin{lemma}
\label{lm:runtime:sketch}
The total runtime of the above algorithm, across all the rounds, is at most $\tilde{O}(m)$.
\end{lemma}

\begin{proof}(Sketch)
Immediately after we invoke the degree-reduction subroutine (see \Cref{lem:degree-reduction}) on some $\calM$-atom (say) $C$, the $\calM$-molecule containing $C$ gets deleted from $\calM$. Thus the total time spent on all the calls to the degree reduction subroutine, across all the rounds, is at most $\tilde{O}(m)$. Moreover, it is easy to ensure that the total time spent on identifying the augmenting edges across all the rounds is bounded by $\tilde{O}(m)$: This is because if an edge $(u, v) \in E$ is {\em not} augmenting at the start of a given round, then it can never become augmenting at the start of any subsequent round. Thus, all we need to do is keep scanning through the edges in $E$ in an arbitrary order, while checking whether or not the current edge being scanned is an augmenting edge. If yes, then we update $\calF$ and $\calM$ accordingly and move to the next round by continuing our scan from the position next to where we left off.
\end{proof}

\subsubsection{Proof of \Cref{lm:rounds:sketch}}
\label{sec:lm:rounds:sketch}
Fix any integer $r$ which satisfies 
\begin{equation}
\label{eq:r}
0 \leq r < f^2/(1000n).
\end{equation}
Suppose that the algorithm has  implemented the first $r$ rounds. We will show that at the end of round $r$, there must necessarily exist at least one augmenting edge w.r.t.~$(\calF, \calM)$, and so the algorithm will also be successful in implementing the next round $r+1$. 
This will conclude the proof of \Cref{lm:rounds:sketch}.

\paragraph{Basic Notations.} We use the superscript $r$ to denote the state of an object at the end of round $r$.  
Accordingly, let $\calF^{(r)}$ and $\calM^{(r)}$ respectively denote the state of  $\calF$ and  $\calM$ at the end of round $r$, and let $N^{(r)} \subseteq V$ denote the set of all $\calM^{(r)}$-non-reducible nodes. 
Furthermore, let $B^{(r)}$ denote the set of all $\calM^{(r)}$-free nodes $v$ with $\deg_{\calF^{(r)}}(v) = \Delta^\star+1$, and let $E_B^{(r)} \subseteq E\left(\calF^{(r)}\right)$ denote the set of all edges in $\calF^{(r)}$ that are incident on at least one node in $B^{(r)}$. 
We define $\calF^{(r)}_{\text{mol}} \subseteq \calF^{(r)}$ to be a subgraph of $\calF^{(r)}$, which consists of the union of all the $\calM^{(r)}$-molecules.\footnote{Recall that each $\calM^{(r)}$-molecule is a component of $\calF^{(r)}$, but not vice versa.} 
Thus, we have 
$$V\left( \calF^{(r)}_{\text{mol}}\right) = \bigcup_{M \in \calM^{(r)}} V(M) \text{ and } E\left( \calF^{(r)}_{\text{mol}} \right) = \bigcup_{M \in \calM^{(r)}} E(M).$$ 
Finally, we define $\calF_{\text{free}}^{(r)} \subseteq \calF^{(r)}$ to be a subgraph of $\calF^{(r)}$, which consists of all the components of $\calF^{(r)}$ that are {\em not} part of $\calF^{(r)}_{\text{mol}}$. 
Thus we have $V\left(\calF_{\text{free}}^{(r)}\right) =  V - V\left(\calF_{\text{mol}}^{(r)}\right)$ and $E\left(\calF_{\text{free}}^{(r)}\right) =  E(\calF) - E\left(\calF_{\text{mol}}^{(r)}\right)$.
Clearly, the set $V\left(\calF_{\text{free}}^{(r)}\right)$ consists of all the $\calM^{(r)}$-free nodes. 
It is easy to verify that
\begin{equation}
\label{eq:verify}
N^{(r)} \subseteq V\left(\calF^{(r)}_{\text{mol}} \right), B^{(r)} \subseteq V\left(\calF_{\text{free}}^{(r)}\right) \text{ and } E_B^{(r)} \subseteq E\left(\calF_{\text{free}}^{(r)}\right).
\end{equation}

\paragraph{Bounding the Number of Edges in $E_B^{(r)}$.} Each round of the algorithm merges two components of $\calF$ into one, and then ensures that each node in the two merged components become $\calM$-free. 
Thus, under \Cref{assumption:uniform-size}, the number of $\calM$-free nodes increases by at most $2 \cdot 100 \cdot n/f \leq 200 n/f$ per round. Thus, from \Cref{eq:verify}, we infer that $\left| E_B^{(r)} \right| \leq \left| E\left(\calF_{\text{free}}^{(r)}\right) \right| \leq \left| V\left(\calF_{\text{free}}^{(r)}\right) \right| - 1 \leq 200 nr/f$; the second inequality holds because $\calF^{(r)}_{\text{free}}$ is a subgraph of the forest $\calF^{(r)}$. To summarize, we have 
\begin{equation}
\label{eq:free:100}
200 nr/f \geq  \left|E_B^{(r)} \right|.
\end{equation}

\paragraph{The Main Argument.} If we delete all the nodes in  $N^{(r)}$ from $\calF^{(r)}_{\text{mol}}$ (see \Cref{eq:verify}), then every component in the resulting forest would be an $\calM^{(r)}$-atom. 
Now, consider the collection $\calC^{(r)}$, which consists of all the $\calM^{(r)}$-atoms of $\calF^{(r)}$, {\em plus} all the nodes in $N^{(r)}$.
Observe that $\calF^{(r)}_{\text{mol}}$ consists of at least $f-2r$ components, as each round reduces the number of  $\calM$-molecules by at most $2$.
Since $\deg_{\calF^{(r)}}(u) = \deg_{\calF^{(r)}_{\text{mol}}}(u) = \Delta^\star+1$ for all $u \in N^{(r)}$, we infer that  
\begin{equation}
\label{eq:collection:100}
\left|\calC^{(r)} \right| \geq f - 2r + \Delta^\star \cdot \left| N^{(r)} \right|.
\end{equation}

Now, for the sake of contradiction, suppose that there is no augmenting edge w.r.t.~$\left(\calF^{(r)}, \calM^{(r)}\right)$. In other words: (i) there is no edge $(u,v) \in E$ where $u$ and $v$ belong to two distinct $\calM^{(r)}$-atoms in two distinct components of $\calF^{(r)}$, and (ii) there is no edge $(u,v) \in E$ where $u$ is in an $\calM^{(r)}$-atom and $v$ is $\calM^{(r)}$-free with $\deg_{\calF^{(r)}}(v) \leq \Delta^\star$. Also, by definition, there is no edge in $E$ which connects two distinct $\calM^{(r)}$-atoms within the same component of $\calF^{(r)}$. Overall, this leads us to the following observation.

\begin{observation}
\label{obs:ex-abs}
Consider any edge $(u, v) \in E$ such that $u$ belongs to an $\calM^{(r)}$-atom (say) $C_u$ and $v \notin V(C_u)$. Then, it must be the case that  $v \in N^{(r)} \cup B^{(r)}$.
\end{observation}

Now, consider the optimal spanning tree $T^\star$ of $G$, which has maximum degree $\Delta^\star$. Clearly, this spanning tree $T^\star$ connects the $\calM^{(r)}$-atoms and the nodes in $N^{(r)}$ together with the rest of the nodes of $G$. Thus, there are at least $\left|\calC^{(r)}\right| - 1$ edges in $T^\star$ that have one endpoint in an $x \in \calC^{(r)}$ ($x$ is either a node in $N^{(r)}$ or an $\calM^{(r)}$-atom) and the other endpoint outside of $x$.  \Cref{obs:ex-abs} implies that each such edge is incident on  $N^{(r)} \cup  B^{(r)}$. Thus, from \Cref{eq:collection:100}, we get
\begin{equation*}
 \Delta^\star \cdot \left|N^{(r)} \right| + \Delta^\star \cdot \left| B^{(r)} \right| \geq \left|\calC^{(r)} \right| - 1 \geq f -2r + \Delta^\star \cdot \left|N^{(r)} \right| - 1.
\end{equation*}
Since $f \gg 1$ (see \Cref{assumption:uniform-size}), rearranging the terms in the above inequality, we infer that
\begin{equation}
\label{eq:collection101}
\Delta^\star \cdot \left| B^{(r)} \right| + 2r \geq f/2.
\end{equation}
By definition, each node in $B^{(r)}$ has $(\Delta^\star+1)$ many incident edges in   $\calF^{(r)}$. Thus, we get
\begin{equation}
\label{eq:collection102}
2 \cdot \left| E_B^{(r)} \right| \geq \sum_{v \in B^{(r)}} \deg_{\calF^{(r)}}(v) \geq \Delta^\star \cdot \left| B^{(r)} \right|. 
\end{equation}
Since $n/f \geq 1$, from \Cref{eq:free:100}, \Cref{eq:collection101} and \Cref{eq:collection102}, we now infer that
\begin{eqnarray*}
500 nr/f \geq 400 nr/f + 2r  \geq 2 \cdot \left| E_B^{(r)} \right| + 2r \geq \Delta^\star \cdot \left| B^{(r)} \right| + 2r \geq f/2.
\end{eqnarray*}

Rearranging the terms, we now get $r \geq f^2/(1000 n)$, which contradicts \Cref{eq:r}. Thus, our assumption that there is no augmenting edge w.r.t.~$\left(\calF^{(r)}, \calM^{(r)} \right)$ must be wrong. This concludes the proof of \Cref{lm:rounds:sketch}.

\subsection{Overview of Our Algorithm (Without Any Simplifying Assumption)}\label{sec:technical-overview-overcome-assumption}

Let us highlight three crucial properties that we used to derive the algorithm for \Cref{lem:simple-alg-uniform-size}. 
\begin{enumerate}
\item Under \Cref{assumption:uniform-size}, the process of degree reduction of an $\calM$-reducible node $u$ (see \Cref{lem:degree-reduction}) only affects $O(n/f)$ nodes of the forest $\calF$, as the component of $\calF$ containing $u$ has size  $O(n/f)$.
\item We compute a {\em blocking set} of augmenting edges in $\tilde{O}(m)$ time (see \Cref{lm:rounds:sketch}). Specifically, there does {\em not} exist any augmenting edge at the end of the last round of the algorithm for \Cref{lem:simple-alg-uniform-size}.
This is reminiscent of the classical blocking flow algorithm for the maxflow problem.
\item Any {\em arbitrary}  blocking set of augmenting edges (not necessarily the one returned by the algorithm) is of size $\Omega(f^2/n)$.
This follows from \Cref{lm:rounds:sketch}.  
\end{enumerate}
In the remainder of this section, we present a high-level overview of the main challenges that arise while attempting to ensure these properties in the general case (when \Cref{assumption:uniform-size} might not hold), and explain how we overcome these challenges.

\medskip
\noindent 
\textbf{Challenge I: Non-Uniform Components.} In the absence of \Cref{assumption:uniform-size}, we need to deal with the fact that some components of the initial forest $\calF$ might not be of size $O(n/f)$.
As a result, a single call to the degree-reduction subroutine (see \Cref{lem:degree-reduction})  might even affect $\Omega(n)$ nodes if the concerned component of $\calF$ happens to  be of size $\Omega(n)$.

To address this challenge, we chop up  every ``large'' (of size $\Omega(n/f)$) component $C$ of $\calF$ into a collection of mutually node-disjoint ``small'' molecules, each of size $O(n/f)$.
We call the resulting structure  a {\bf $\theta$-molecular decomposition} (see \Cref{sec:algorithm}), with $\theta = O(n/f)$.
We work with the atoms  w.r.t.~these small molecules.
Thus, even if a component $C$ of $\calF$ is large, when we reduce the degree of a node $u$ in $C$, it only affects at most $O(n/f)$ nodes in the molecule containing $u$, and it remains possible in future to reduce the degrees of other nodes that are in different small molecules of $C$.

\medskip
\noindent 
\textbf{Challenge II: Size of a Blocking Set of Augmenting Edges.}
Unfortunately,  we can no longer lower  bound the size of a blocking set of augmenting edges as in \Cref{lm:rounds:sketch}, for the following reason.
It might happen that a node $v$ in a component $C$ of $\calF$ is reducible w.r.t.~the entire component $C$, but is {\em not} reducible w.r.t.~the concerned small molecule of size $O(n/f)$ which contains $v$ in our $\theta$-molecular decomposition.
As a result, lots of augmenting edges (as per \Cref{def:augmenting}) can no longer be considered as augmenting.

To address this challenge, we define an intricate object called an {\bf augmenting chain} (see \Cref{def:alt-chain} and \Cref{def:aug-chain}). To get some intuition behind this definition, consider a non-forest edge $(u, v) \in E - E(\calF)$ between an $\calM$-reducible node $u$, and an $\calM$-non-reducible node $v$ which  lies outside of the component of $u$.
Let us refer to such an edge $(u, v)$ as a \textit{bad} edge. Note that if 
 we encounter a bad edge $(u,v)$, then  we can try to perform the following operations.

\medskip
\noindent 
\textbf{Tiny Swap:}
\begin{enumerate}
\item Insert the edge $(u,v)$ into the forest $\calF$. \item Consider a  sub-tree $T^{\calF}_{xv}$ of $\calF$, where $(x, v) \in E(\calF)$. Delete the edge $(x,v)$ from $\calF$, so that  the degree of $v$ remains $\leq \Delta^\star+1$ and the tree $T^{\calF}_{xv}$ gets disconnected from the rest of the forest $\calF$.
\item Try to search for edges that can reconnect $T^{\calF}_{xv}$ back to the rest of the forest $\calF$.
\end{enumerate}

Morally, an augmenting chain consists of a sequence of tiny swaps as described above, until we find an edge connecting two molecules which is an augmenting edge.

We derive an analogue of  \Cref{lm:rounds:sketch}, by considering augmenting chains of length $O(n/f)$.
More precisely, we show that any {\em blocking set} of  $O(n/f)$-length augmenting chains has size at least $\Omega(f^3/n^2)$ (see \Cref{key-lemma}). 
Here, the phrase ``blocking set'' refers to the property that we can {\em not} find any augmenting chain of length $O(n/f)$, after we have successively applied all the  chains in to the concerned set.
Note that this bound is slightly weaker than \Cref{lm:rounds:sketch}. This is because applying an augmenting chain of length $\Omega(n/f)$ can affect $\Omega(n^2/f^2)$ nodes, but applying an augmenting edge only affects $O(n/f)$ nodes.
But this weaker guarantee is still sufficient for us to break the $O(mn)$ time barrier.

\medskip
\noindent 
\textbf{Challenge III: Finding a Blocking Set of Augmenting Chains Algorithmically.}
Our definition of an augmenting chain is quite complex, and searching for such objects is non-trivial, especially since we need to find a blocking set of them within a limited time.
Moreover, after improving the forest $\calF$ by some tiny swaps, the validity of the remaining augmenting chains may be cast into doubt.

We provide an intricate algorithm (see \Cref{sec:description-subroutine-main}) that searches for augmenting chains by exploring potential tiny swaps sequentially for $n/f$ layers, until it finds augmenting chains of length $\simeq n/f$.
The algorithm is involved, since we do not know in advance which layer of tiny swaps an edge $(u, v)$ might contribute to.
Along the way, we need to prove some interesting structural properties of augmenting chains.
The resulting procedure which addresses Challenge III runs in  $\tilde{O}(mn/f)$ time, where the extra $n/f$ factor overhead in runtime (in comparison with \Cref{lem:simple-alg-uniform-size}) arises because the augmenting chains are now of length $O(n/f)$. 

\medskip
\noindent
\textbf{Putting Everything Together.}
As per the above discussions, we bypass  \Cref{assumption:uniform-size}, and provide an algorithm with guarantees analogous to that of \Cref{lem:simple-alg-uniform-size}. In $\tilde{O}(mn/f)$ time, this allows us to reduce the number of components of the forest $\calF$ by $\Omega(f^3/n^2)$.
Overall, this leads to an additive one approximate minimum degree spanning tree algorithm that runs in $\tilde{O}(mn^{3/4})$ time  (see \Cref{sec:algorithm}).

\section{Alternating and Augmenting Chains}\label{sec:alt-aug-chains}

Our definitions of alternating and augmenting chains are a bit intricate. To ease into them, we start with a few helpful terminologies. Consider any {\em valid} forest $\calF$, and any molecular decomposition $\calM$ of $\calF$. For every pair of distinct nodes $x, y \in V$ within the same component of $\calF$, we refer to  $T_{x \leftarrow y}^{\calF}$ as an {\bf $\calM$-block} iff the following two conditions hold:
\begin{itemize}
    \item $y$ is either $\calM$-non-reducible, or the root of a normal $\calM$-molecule.
    \item $T^\calF_{x \leftarrow y}$ is contained in an $\calM$-molecule.
\end{itemize}

\begin{observation}
\label{obs:block}
   Consider any molecular decomposition $\calM$ of a valid forest $\calF$. For every $\calM$-atom $C$ and $\calM$-block $T_{x \leftarrow y}^\calF$, either $V(C) \subseteq V\left(T_{x \leftarrow y}^\calF\right)$ or $V(C) \cap V\left(T_{x \leftarrow y}^\calF\right) = \emptyset$. In other words, an $\calM$-block cannot have a non-trivial intersection with an $\calM$-atom.
\end{observation}
\begin{proof}
    If the node $y$ is the root of a normal $\calM$-molecule, then the observation is trivial.
    Thus, for the rest of the proof, suppose that the node $y$ is $\calM$-non-reducible and belongs to some molecule $T$ of $\calM$.
    Each $\calM$-atom in $T$ is a connected sub-tree of $T$, and any two distinct $\calM$-atoms are mutually node-disjoint. By definition, the node $y$ is not in any of these $\calM$-atoms.
    This implies that $T^\calF_{x \leftarrow y}$ can not have a non-trivial intersection with any $\calM$-atom.
\end{proof}

Throughout the rest of this paper, we say that a triple $(\calF, \calM, D)$ is a {\bf valid configuration} iff $\calF$ is a {\em valid} forest, $\calM$ is a molecular decomposition of $\calF$, and $D$ is a {\em subset} of $\calM$-free nodes satisfying the following condition: Every node $u \in D$ has $\deg_{\calF}(u) = \Delta^\star$. We refer to the nodes in $D$ as {\bf dirty} w.r.t.~$(\calF, \calM, D)$.
We are now ready to define the notion of an alternating chain.

\begin{definition}[Alternating Chain]\label{def:alt-chain}
Let $(\calF, \calM, D)$ be a valid configuration and $\ell \geq 0$\footnote{In the case $\ell=0$, the sequence $P$ consists only of $(w_0)$, and the only conditions that matter is \Cref{property-z0}, and the definition of $0^\th$ $\calM$-block of $P$.} be an integer.
We call a sequence of {\em distinct} nodes $P = (w_0, z_1, w_1, z_2, w_2, \ldots, z_\ell, w_\ell)$  an {\bf alternating chain} of {\bf length} $\ell$ w.r.t.~$(\calF, \calM, D)$ iff  the following properties hold.
\begin{enumerate}
    \item\label{property-z0}
    $w_0$ belongs to an $\calM$-atom within a special $\calM$-molecule.

    \item\label{property-wi-zi-connection}
    For every $i \in [1, \ell]$:
    \begin{itemize}
    \item{(a)} $w_i$ and $z_i$ are in the same connected component of $\calF$.

    \item{(b)} $T_{w_i \leftarrow z_i}^\calF$ is an $\calM$-block.
    
    \item{(c)} $w_i$ belongs to an $\calM$-atom.

  \end{itemize}

 (For all $i \in [0, \ell]$, we refer to the  $\calM$-atom containing $w_i$ as the {\bf $i^\th$ critical  $\calM$-atom of $P$}. Further, for all $i \in [1, \ell]$, we refer to $T_{w_i \leftarrow z_i}^\calF$ as the {\bf $i^{\th}$ $\calM$-block of $P$}. Finally, for ease of exposition, we refer to the special $\calM$-molecule which contains $w_0$ as the {\bf $0^{\th}$ $\calM$-block of $P$}.)

    \item\label{property-wi-outside-jth-block} For all $i, j \in [0, \ell]$ with $i > j$, the node $w_i$ lies outside of the  $j^\th$ $\calM$-block.

    \item\label{property-non-forest-edge} For all $i \in [1, \ell]$, we have $(w_{i-1}, z_i) \in E- E(\calF)$, i.e., $(w_{i-1}, z_i)$ is a non-forest edge.
\end{enumerate}
\end{definition}

It is easy to see that the set of all $\calM$-blocks has a laminar structure, especially, the $\calM$-blocks of an alternating chain satisfy the property as summarized in the observation below. We will use this observation throughout the analysis of our algorithm. 
\begin{observation}
\label{obs:key} Let $P$ be an alternating chain of length $\ell$ w.r.t.~a valid configuration $(\calF, \calM, D)$. 
Then, for $0 \leq j < i \leq \ell$, either the node-sets of $j^\th$ and $i^\th$ $\calM$-blocks of $P$  are mutually disjoint, or the $j^\th$ $\calM$-block is a subgraph of the $i^\th$ $\calM$-block.
\end{observation}

\begin{proof}
By definition, the node $w_i$ belongs to an $\calM$-atom and lies outside the $j^{th}$ $\calM$-block of $P$. 
Thus, by \Cref{obs:block},  every node in the $\calM$-atom containing $w_i$ lies outside the $j^{th}$ $\calM$-block of $P$. At this point, the observation follows from the definition of an $\calM$-block.
\end{proof}

We are now ready to define the concept of an augmenting chain, which is an alternating chain appended by one extra node with certain properties.

\begin{definition}[Augmenting Chain]\label{def:aug-chain}
Let $(\calF, \calM, D)$ be a valid configuration and $\ell \geq 0$ be an integer.
We call a sequence of {\em distinct} nodes $A = \left(w_0, z_1, w_1, z_2, w_2, \ldots, z_\ell, w_\ell, z_{\ell+1} \right)$ an {\bf augmenting chain} of {\bf length} $\ell+1$ w.r.t.~$(\calF, \calM, D)$ iff the following properties hold. 
\begin{enumerate}
    \item\label{property1:aug-chain} The prefix $P(A) = \left(w_0, z_1, w_1, z_2, w_2, \ldots, z_\ell, w_\ell \right)$ of $A$ is an alternating chain w.r.t.~$(\calF, \calM, D)$. 
    
    (For each $i \in [0, \ell]$, we use the terms ``{\bf $i^{th}$ $\calM$-block of $A$}'' and ``{\bf $i^{th}$ critical $\calM$-atom of $A$}'' to refer to the ``$i^{th}$ $\calM$-block of $P(A)$'' and ``$i^{th}$ critical $\calM$-atom of $P(A)$'', respectively.) 

    \item\label{property-zell1-outside-blocks}
    For every $i \in [0, \ell]$, the node $z_{\ell+1}$ lies outside the $i^\th$ $\calM$-block of $A$. 

 \item\label{property-well-zell1}
    $(w_\ell, z_{\ell+1}) \in E - E(\calF)$, i.e., $(w_\ell, z_{\ell+1})$ is a non-forest edge.

    \item \label{property-zell1-reducible}
    $z_{\ell+1}$ is either $\calM$-reducible, or  an $\calM$-free node with $\deg_{\calF}(z_{\ell+1}) \leq \Delta^\star$ that lies outside  $D$.
\end{enumerate}
\end{definition}

\Cref{fig:aug-chain} illustrates an example of an augmenting chain of length $8$.
Consider any valid configuration $(\calF, \calM, D)$, and let 
  $A = \left(w_0, z_1, w_1, z_2, w_2, \ldots, z_\ell, w_\ell, z_{\ell+1} \right)$ be an augmenting chain w.r.t.~$(\calF, \calM, D)$. We use the phrase ``{\bf applying the augmenting chain $A$}'' to refer to an operation which gives us a new valid configuration $(\calF', \calM', D')$, defined as follows.

\begin{itemize}
    \item For every $i \in [1, \ell]$, let $(z_i, y_i) \in E(\calF)$ be the edge connecting $z_i$ to the $i^\th$ $\calM$-block of $A$,\footnote{In other words, $y_i$ is the first node after $z_i$ on the path $P_{w_i \leftarrow z_i}^{\calF}$ from $z_i$ to $w_i$ in $\calF$.} just {\em before} we apply the augmenting chain $A$.
    \item \textbf{Updating $\calF$}. We perform the following steps.
    \begin{enumerate}
    \item \label{step:1} First, for every $i \in [0, \ell]$, we reduce $\deg_{\calF}(w_i)$ to $\leq \Delta^\star$ by applying \Cref{lem:degree-reduction} in the $i^\th$ critical $\calM$-atom of $A$. Moreover, if $z_{\ell+1}$ is $\calM$-reducible, then we also reduce $\deg_{\calF}(z_{\ell+1})$ to $\leq \Delta^\star$ by applying \Cref{lem:degree-reduction} in the $\calM$-atom which contains $z_{\ell+1}$ (Otherwise, by definition, we already have that $z_{\ell+1}$ is an $\calM$-free node with $\deg_{\calF}(z_{\ell+1}) \leq \Delta^\star$). 
    \item \label{step:2} Next, for every $i \in [1, \ell]$, we delete the edge $(z_i, y_i)$ from $\calF$. 
    \item \label{step:3} Finally, for every $i \in [1, \ell+1]$, we insert the edge $(w_{i-1}, z_{i})$ into $\calF$. 
    \end{enumerate}
    We let $\calF'$ be the status of $\calF$ after performing the above three steps. 
    \item \textbf{Updating $\calM$.}
    If a molecule $M \in \calM$ contains at least one node from $\{w_0,w_1, \ldots, w_\ell, z_{\ell+1}\}$ just {\em before} applying $A$, then we say that the molecule is {\bf affected} due to applying $A$. We delete all affected molecules from  $\calM$. We let $\calM'$ be the status of $\calM$ after making these changes. 
    \item \textbf{Updating $D$.} Let $Y$ denote the set of nodes $y \in \{y_1, \ldots, y_{\ell}\}$ that were $\calM$-non-reducible just {\em before} applying the augmenting chain $A$. We set $D' \leftarrow D \cup Y$. 
\end{itemize}

\input{figures/fig-aug-chain}

\begin{lemma}\label{lem:apply-aug-chain}
    After applying an augmenting chain $A$ w.r.t.~a valid configuration $(\calF, \calM, D)$, the result $(\calF', \calM', D')$ is another valid configuration where the forest $\calF'$ has one fewer component than $\calF$.
\end{lemma}

\begin{proof}
    Recall the procedure for updating $\calF$, as described above. Just before applying this procedure, the $\ell+2$ different $\calM$-atoms which respectively contain  $w_0, \ldots, w_{\ell}, z_{\ell+1}$ are all mutually node-disjoint (see \Cref{def:aug-chain}), and each $\calM$-block of $A$ has a trivial intersection with each of these $\calM$-atoms (see \Cref{obs:block}). Thus, during Step~\ref{step:1} the value of $\deg_{\calF}(z_i)$ remains unchanged for every $i \in [1, \ell]$. We conclude that $\calF$ is still a valid forest at the end of Step~\ref{step:1}, with the additional property that: 
    \begin{equation}
    \label{eq:add}
    \deg_{\calF}(w_i) \leq \Delta^\star \text{ for all } i \in [0, \ell], \text{ and } \deg_{\calF}(z_{\ell+1}) \leq \Delta^\star.
    \end{equation}
    Each edge-deletion $(z_i,y_i)$ in Step~\ref{step:2} {\em increases} the number of components in $\calF$ by one, and detaches the $i^\th$ $\calM$-block $T^{\calF}_{w_i \leftarrow z_i}$ from the rest of the forest.
    Thus, due to Step~\ref{step:2}, the number of components in $\calF$ increases by $\ell$.
    In contrast, each edge-insertion $(w_{i-1}, z_i)$ in Step~\ref{step:3} {\em decreases} the number of components in $\calF$ by one, and reattaches the $\calM$-block $T^{\calF}_{w_i \leftarrow z_i}$ to the forest.
    Thus, due to Step~\ref{step:3}, the number of components in $\calF$ decreases by $\ell+1$.
    We infer that overall the number of components in $\calF$ decreases by $(\ell+1)-\ell = 1$ due to the above procedure, and so the forest $\calF'$ has one fewer component than $\calF$.
    Note that, during this process, $\calF$ remains a forest and no cycle in formed due to the following crucial facts: Just before applying $A$, \Cref{obs:key} holds, and the node $z_{\ell+1}$ lies outside the $i^{th}$ $\calM$-block of $A$ for every $i \in [0, \ell]$ (see \Cref{def:aug-chain}). 

    Moreover, due to Step~\ref{step:2} and Step~\ref{step:3}, the nodes in $\{w_0, w_1, \ldots, w_{\ell}, z_{\ell+1}\}$ increase their degrees in $\calF$ by one; the nodes in $\{y_1, \cdots, y_\ell\}$ decrease their degrees by one; and the degree of every other node in $\calF$ remains unchanged in these two steps. Accordingly, \Cref{eq:add} implies that $\calF$ is a valid forest at the end of Step~\ref{step:3}.
    To summarize, we have derived that $\calF'$ is a valid forest with one fewer component than $\calF$.

    Next, we focus on the procedure for updating $\calM$. Recall that $\calM'$ is obtained by removing all the affected molecules from $\calM$. Accordingly, every molecule $M \in \calM'$ must necessarily be a molecule in both the forests $\calF$ (before applying $A$) and $\calF'$ (after applying $A$). As $\calM' \subseteq \calM$ and $\calM$ is a molecular decomposition of $\calF$, it follows that $\calM'$ is also a molecular decomposition of $\calF'$. 
    
    Finally, we focus on the procedure for updating $D$.
    Consider any node $ y_i \in Y$. By definition, just before applying $A$, we  had $\deg_{\calF}(y_i) = \Delta^{\star}+1$ since $y_i$ was $\calM$-non-reducible. While updating $\calF$,  we delete the edge $(z_i, y_i)$, and no other edge incident on $y_i$ gets inserted.
    It follows that $\deg_{\calF'}(y_i) = \Delta^\star$.
    Moreover, the node $y_i$ belongs to an affected molecule, and hence it is $\calM'$-free.
    Next, consider any node $v$ that  was part of the set $D$ (which implies that $v$ was $\calM$-free and had $\deg_{\calF}(v) = \Delta^\star$) just before applying $A$. It is also easy to verify that $v$ remains $\calM'$-free, and no edge incident on $v$ gets inserted/deleted as we update the forest $\calF$. So, we have $\deg_{\calF'}(v) = \Delta^\star$ for all $v \in D$. To summarize, the set $D'$ continues to satisfy the desired conditions that: (i)  it is a subset of $\calM'$-free nodes and (ii) every node $u \in D'$ has $\deg_{\calF'}(u) = \Delta^\star$. 

    From the above discussion, we infer that $(\calF', \calM', D')$ is also a valid configuration and $\calF'$ has one fewer component than $\calF$. This concludes the proof of the lemma.
\end{proof}

\Cref{fig:aug-chain-apply} illustrates the result of applying the augmenting chain in \Cref{fig:aug-chain}.

\input{figures/fig-aug-chain-apply}

We next consider applying multiple augmenting chains one after another. Let $(\calF, \calM, D)$ be any valid configuration. We say that $A^{(1)}, \ldots, A^{(k)}$  is a {\bf sequence of augmenting chains} w.r.t.~$(\calF, \calM, D)$ iff  $A^{(i)}$ is an augmenting chain w.r.t.~$(\calF^{(i-1)}, \calM^{(i-1)}, D^{(i-1)})$ for all $i \in [1, k]$, where $(\calF^{(i)}, \calM^{(i)}, , D^{(i)})$ is the outcome of applying $A^{(i)}$ on $(\calF^{(i-1)}, \calM^{(i-1)}, , D^{(i-1)})$ and $(\calF^{(0)}, \calM^{(0)}, D^{(0)}) := (\calF, \calM, D)$.

We conclude this section by highlighting an important property of augmenting chains. 
In particular, the primary reason why we need to consider the set of dirty nodes $D$ while defining a valid configuration $(\calF, \calM, D)$  is that this allows us to derive \Cref{claim:set-of-augmenting-chains-is-decreasing}, as stated below.

\begin{claim}
\label{claim:set-of-augmenting-chains-is-decreasing}
    Assume $A$ is an augmenting chain w.r.t.~a valid configuration $(\calF, \calM, D)$, and $(\calF', \calM', D')$ is the result of applying $A$ to $(\calF, \calM, D)$.
    Then, any augmenting chain w.r.t.~$(\calF', \calM', D')$ is an augmenting chain w.r.t.~$(\calF, \calM, D)$ as well.
\end{claim}

\begin{proof}
    Assume $A = (w_0^A, z_1^A, w_1^A, \ldots, z_k^A, w_k^A, z_{k+1}^A)$.
    Let $C = (w_0, z_1, w_1, \ldots, z_\ell, w_\ell, z_{\ell+1}) $ be an arbitrary augmenting chain w.r.t.~$(\calF', \calM', D')$.
    We show it is also an augmenting chain w.r.t.~$(\calF, \calM, D)$.
    In the following, we show the properties of \Cref{def:aug-chain} for $C$ w.r.t.~$(\calF, \calM, D)$.
    For every $i \in [0, \ell]$, we have $(w_i, z_{i+1}) \in E -  E(\calF')$, and $w_i$ is contained in an $\calM'$-atom.
    We conclude that the $\calM'$-molecule containing $w_i$ (which is an $\calM$-molecule as well) has remained unaffected while applying $A$ to $(\calF, \calM, D)$.
    This means that the neighborhood of $w_i$ coincide in both $\calF$ and $\calF'$, concluding that $(w_i, z_{i+1}) \in E -  E(\calF)$.
    The same statement is correct for all the blocks $T^{\calF'}_{w_i \leftarrow z_i}$, i.e., $T^{\calF'}_{w_i \leftarrow z_i} = T^{\calF}_{w_i \leftarrow z_i}$. 
    This fact combined with $\calM' \subseteq \calM$ clearly concludes properties \ref{property1:aug-chain}, \ref{property-zell1-outside-blocks}, and \ref{property-well-zell1}.
    It remains to show property \ref{property-zell1-reducible}.
    There are two cases for $z_{\ell+1}$ as follows.

    \medskip
    \noindent
    \textbf{Case I.}
    $z_{\ell+1}$ is $\calM'$-reducible. 
    According to $\calM' \subseteq \calM$, we obviously have that $z_{\ell+1}$ is $\calM$-reducible as well, which concludes property \ref{property-zell1-reducible} for $C$ w.r.t.~$(\calF, \calM, D)$.

    \medskip
    \noindent
    \textbf{Case II.}
    $z_{\ell+1}$ is $\calM'$-free with $\deg_{\calF'}(z_{\ell+1}) \leq \Delta^\star$ and $z_{\ell+1} \notin D'$.
    We consider two sub-cases as follows.
    \begin{itemize}
        \item $z_{\ell+1}$ is $\calM$-free.
        According to the operation of applying the augmenting chain $A$ to $(\calF, \calM, D)$, it is straightforward to see that the degree of $\calM$-free nodes can only increase (the only $\calM$-free node whose degree can increase is the last node $z^A_{k+1}$ in $A$).
        We conclude that $\deg_{\calF}(z_{\ell+1}) \leq \deg_{\calF'}(z_{\ell+1}) \leq \Delta^\star$, and $z_{\ell+1} \notin D$ since $D' \supseteq D$.
        Hence, we have property \ref{property-zell1-reducible} for $C$ w.r.t.~$(\calF, \calM, D)$.
        
        \item $z_{\ell+1}$ is $\calM$-covered.
        If $z_{\ell+1}$ is $\calM$-reducible, we already have property \ref{property-zell1-reducible} for $C$ w.r.t.~$(\calF, \calM, D)$.
        Otherwise, $z_{\ell+1}$ is $\calM$-non-reducible, which concludes $\deg_{\calF}(z_{\ell+1}) = \Delta^\star+1$.
        We also have $\deg_{\calF'}(z_{\ell+1}) \leq \Delta^\star$.
        The degree-reduction subroutine for $w_i^A$ (and possibly $z^A_{k+1}$) only affects $\calM$-atoms and does not change the degree of $z_{\ell+1}$.
        Hence, during applying $A$ to $(\calF, \calM, D)$ the only scenario that can happen which decreases the degree of $z_{\ell+1}$ is the following; $z_{\ell+1} = y_i^A$ is the first node after $z^A_i$ in the path $P^{\calF}_{z_i^A, w_i^A}$ from $z_i^A$ to $w_i^A$ in $\calF$, and the edge $(z^A_i,y^A_i) = (z_i^A, z_{\ell+1})$ is removed from $\calF$.
        Finally, $z_{\ell+1}$ satisfies all the properties while defining a dirty node which means $z_{\ell+1} \in D'$.\qedhere
    \end{itemize}
\end{proof}

\section{Our Algorithm}\label{sec:algorithm}

Our algorithm works by repeatedly applying a sequence of augmenting chains to  reduce the number of components in a valid forest by a large additive factor. Before describing the algorithm, we need to introduce a couple of  concepts -- (i) ``$t$-configuration'' and (ii) ``$\theta$-molecular decomposition''. 

We say that a valid configuration $(\calF, \calM, D)$ is a {\bf $t$-configuration}, for some integer $t \geq 1$, iff the length of every augmenting chain w.r.t.~$(\calF, \calM, D)$ is at least $t$.
Next, fix any integer $\theta \in [1, n]$, and consider any molecular decomposition $\calM$ of a valid forest $\calF$. We say that $\calM$ is a {\bf $\theta$-molecular decomposition} of $\calF$ iff the following conditions hold.
\begin{enumerate}
\item\label{property:theta-max-decomposition-1} Every component of $\calF$ with at most $2\theta$ nodes is a special $\calM$-molecule.

\item\label{property:theta-max-decomposition-2} For every normal $\calM$-molecule $M$, we have $|V(M)| \leq \theta$;

\item\label{property:theta-max-decomposition-3}
For any arbitrary molecule $T^{\calF}_{x \leftarrow y}$ of $\calF$ which is not part of $\calM$ and $T^{\calF}_{x \leftarrow y}$ contains at least one $\calM$-free node, we must have $|V(T^{\calF}_{x \leftarrow y})| \geq \theta+1$.
\end{enumerate}

Intuitively, we can compute a $\theta$-molecular decomposition $\calM$ of $\calF$ in near-linear time in a bottom-up manner, starting from the leaves in $\calF$.
To see how this works, say that a component of $\calF$ is {\em small} if it contains at most $2\theta$ nodes, and {\em big} otherwise.
We start by designating every small component of $\calF$ as a special $\calM$-molecules, and every leaf node in a big component of $\calF$ as a singleton normal $\calM$-molecule.
Subsequently, as long as there exists an edge $(u, v) \in E(\calF)$ in a big component of $\calF$ such that $\left|V\left(T_{uv}^{\calF}\right)\right| \leq \theta$ and $T_{wu}^{\calF}$ is an $\calM$-molecule for all neighbors $w \in \psi_{\calF}(u) \setminus \{v\}$, we delete all the molecules $\left\{T_{wu}^{\calF}\right\}_{w \in \psi_{\calF}(u) \setminus \{v\}}$ from $\calM$ and insert the molecule $T_{uv}^{\calF}$ into $\calM$.
The procedure terminates when no more molecules can be added to $\calM$ in this manner. 
The following figure illustrates and example of a $\theta$-molecular decomposition for $\theta=4$.
\input{figures/fig-theta-decomposition}

\begin{observation}
    The above process, provides a $\theta$-molecular decomposition of $\calF$.
\end{observation}

\begin{proof}
Consider an arbitrary component $C$ of $\calF$.
If $|V(C)| \leq 2\theta$, the algorithm adds $C$ to $\calM$ as an special $\calM$-molecule which is aligned with property \ref{property:theta-max-decomposition-1}.
Now, assume $|V(C)| \geq 2\theta + 1$.

Properties \ref{property:theta-max-decomposition-2} and \ref{property:theta-max-decomposition-3} directly follow from the procedure of the algorithm.
It remains to show that $\calM$ is a molecular decomposition of $\calF$, i.e., $\calM$-molecules are mutually node disjoint and the roots of $\calM$-molecules are $\calM$-free nodes.

We show that at any point in time during the construction of $\calM$-molecules in $C$, for every $\calM$-molecule $T^\calF_{x \leftarrow y}$, $y$ is an $\calM$-free node.
This shows that the final $\calM$ is indeed a molecular decomposition of $\calF$ since $\calM$-molecules are connected subgraphs of $\calF$, and their roots are $\calM$-free nodes which separates their node sets.

This is correct for the initial $\calM$ containing the leaves of $C$.
Now, assume that it is violated for the first time while the procedure removes $T^\calF_{w_1 u}, T^\calF_{w_1 u}, \ldots, T^\calF_{w_t u}$ from $\calM$ and adds  $T^\calF_{uv}$ to $\calM$, where $\psi_\calF(u) = \{v,w_1,w_2,\ldots, w_t\}$.
The only scenario that can happen is the following;
$t=1$, and $T^\calF_{uw_1}$ was already in $\calM$.
But, in this case, $V(T^\calF_{vu}) \cup V(T^\calF_{uw_1})$ contains the entire node set of the component $C$ and we have $|V(T^\calF_{vu})| \leq \theta$ and $|V(T^\calF_{uw_1})| \leq \theta$ according to the procedure of the algorithm.
This contradicts the assumption that $|V(C)| \geq 2\theta + 1$.
\end{proof}

We now state two key lemmas that underpin the design and analysis of our algorithm. Say that an augmenting chain is {\em short} iff its length is at most some parameter $H := \lceil 20 n/f \rceil$, where $f$ is the number of components in the current valid forest $\calF$. In words, \Cref{key-lemma} gives us the following guarantee.  If we start with a  valid configuration $(\calF, \calM, D)$, where $D = \emptyset$ and $\calM$ is a $\theta$-molecular decomposition of $\calF$ with $\theta := \lceil 20 n/f\rceil$, then we must apply a sequence of at least $\Omega(f^3/n^2)$  short augmenting chains, before we can arrive at an $(H+1)$-configuration. \Cref{lem:subroutine-main}, in contrast, says that if we start with {\em any} $\ell$-configuration, for $\ell \geq 1$, then in $\tilde{O}(m)$ time we can arrive at an $(\ell+1)$-configuration after applying a sequence of augmenting chains of length $\ell$. We defer the proofs of \Cref{key-lemma} and \Cref{lem:subroutine-main} to \Cref{sec:main-lemma} and \Cref{sec:description-subroutine-main}, respectively.

\begin{lemma}
\label{key-lemma}
Consider any valid forest $\calF^{(0)}$ with $f \geq 20$ components. Define $\theta := \lceil 20 n/f\rceil$ and $H := \lceil 20 n/f \rceil$. Let $\calM^{(0)}$ be a $\theta$-molecular decomposition of $\calF^{(0)}$, and let $D^{(0)} := \emptyset$.
Fix an integer $s \geq 1$. Let $A^{(1)}, \ldots, A^{(s)}$ be a sequence of augmenting chains w.r.t.~$\left(\calF^{(0)}, \calM^{(0)}, D^{(0)}\right)$, such that for all $i \in [1, s]$, the augmenting chain $A^{(i)}$ is of length at most $H$.
For each $i \in [1, s]$, let $\left(\calF^{(i)}, \calM^{(i)}, D^{(i)} \right)$ be the resulting valid configuration we obtain after applying the sequence of augmenting chains $A^{(1)}, \ldots, A^{(i)}$.
Now, if $\left(\calF^{(s)}, \calM^{(s)}, D^{(s)} \right)$ is an $(H+1)$-configuration, then it must be the case that $s \geq f^3/(10^5 \cdot n^2)$.
\end{lemma}

\begin{lemma}\label{lem:subroutine-main}
    There exists a deterministic algorithm that, given an integer $\ell \geq 1$ and an $\ell$-configuration $(\calF, \calM, D)$, returns an $(\ell+1)$-configuration $(\calF', \calM', D')$ by applying a sequence of augmenting chains $A^{(1)},  \ldots, A^{(q)}$ w.r.t.~$(\calF, \calM, D)$. For each $i \in [1, q]$, the augmenting chain $A^{(i)}$ is  of length exactly $\ell$.
    The algorithm runs in $\tilde{O}(m)$ time.
\end{lemma}

\subsection{Description of Our Algorithm}\label{sec:description-alg}

We start by initializing $\calF$ to be a forest with $n$ components and an empty edge-set, i.e. with $E(\calF) \leftarrow \emptyset$ and $V(\calF) \leftarrow V$. Subsequently, our algorithm consists of two stages.

\medskip
\noindent 
\textbf{Stage I.}
Let $f$ denote the number of components of the valid forest $\calF$. 
As long as $f \geq n^{3/4}$, we keep implementing the next {\bf round} of our algorithm. Each round consists of the following steps.
\begin{itemize}
\item Set $H := \lceil 20n/f \rceil$, $\theta = \lceil 20n/f \rceil$, $\calF^{(0)} := \calF$, and $D^{(0)} := \emptyset$.
\item Compute a $\theta$-molecular decomposition $\calM^{(0)}$ of $\calF^{(0)}$. 

(Observe that $(\calF^{(0)}, \calM^{(0)}, D^{(0)})$ is a $1$-configuration.)
\item For every $\ell \in [1, H]$
\begin{itemize}
\item Call the subroutine from \Cref{lem:subroutine-main} on the $\ell$-configuration $\left(\calF^{(\ell-1)}, \calM^{(\ell-1)}, D^{(\ell-1)}\right)$, which returns an $(\ell+1)$-configuration $(\calF^{(\ell)}, \calM^{(\ell)}, D^{(\ell)})$ after applying a sequence of $s_{\ell}$ augmenting chains (each of length $\ell$) w.r.t.~$(\calF^{(\ell-1)}, \calM^{(\ell-1)}, D^{(\ell-1)})$.
\end{itemize}
\item At the end of the above {\bf for loop}, we clearly have an $(H+1)$-configuration $(\calF^{(H)}, \calM^{(H)}, D^{(H)})$, which is derived by applying a sequence of $\sum_{\ell=1}^H s_\ell$ augmenting chains w.r.t.~$(\calF^{(0)}, \calM^{(0)}, D^{(0)})$. We now reset $\calF \leftarrow \calF^{(H)}$, and terminate this round.
\end{itemize}

\medskip
\noindent 
\textbf{Stage II.} We initiate this stage when we observe that the current valid forest $\calF$ has $f < n^{3/4}$ components at the end of a given round. 
We implement this stage in {\bf iterations}. Each iteration reduces the number of components of $\calF$ by one, while ensuring that $\calF$ remains valid, by invoking \Cref{lem:alg-small-scale-main}. 
This stage ends when $\calF$ consists of only one component, and hence is a spanning tree of $G$ with maximum degree at most $\Delta^\star+1$. 
Our algorithm returns this  spanning tree as the output.

\subsection{Analysis of Our Algorithm: Proof of \Cref{theorem:main}}\label{sec:analysis}

From the description of our algorithm, it is obvious that it returns a spanning tree of maximum degree at most $\Delta^\star+1$.
Indeed, 
\Cref{lem:subroutine-main} and \Cref{lem:alg-small-scale-main} guarantee that $\calF$ always remains a {\em valid} forest (i.e., has maximum degree at most $\Delta^\star+1$), and we terminate when $\calF$ has only one component. We devote the rest of this section towards analyzing the runtime of our algorithm.

To begin with, it is easy to see that Stage II of our algorithm  takes $\tilde{O}(mn^{3/4})$ time. This is because Stage II consists of at most $n^{3/4}$ iterations, and each iteration takes $\tilde{O}(m)$ as per \Cref{lem:alg-small-scale-main}. It now remains to show that Stage I of our algorithm also takes at $\tilde{O}(mn^{3/4})$ time.

Consider any given round in Stage I. The round starts with a valid configuration $(\calF^{(0)}, \calM^{(0)}, D^{(0)})$, where $D^{(0)} = \emptyset$ and $\calM^{(0)}$ is a $\theta$-molecular decomposition of $\calF^{(0)}$.
The round ends with an $(H+1)$-configuration $(\calF^{(H)}, \calM^{(H)}, D^{(H)})$, after applying a sequence of $\sum_{\ell = 1}^{H} s_\ell$ augmenting chains (each of length $\leq H$) w.r.t.~$(\calF^{(0)}, \calM^{(0)}, D^{(0)})$. Accordingly, by \Cref{key-lemma}, there are at least $\sum_{\ell = 1}^{H} s_\ell \geq f^3/(10^5 \cdot n^2)$ many augmenting chains in this sequence. In other words, each round reduces the number of components in $\calF$ by at least $f^3/(10^5 \cdot n^2)$. Furthermore, by \Cref{lem:subroutine-main}, each round takes  $\tilde{O}(mH) = \tilde{O}(mn/f)$ time. This holds because the concerned {\bf for loop} in any given round runs for $H$ iterations, and each iteration makes one call to the subroutine guaranteed by \Cref{lem:subroutine-main}.

To summarize,  each round takes $\tilde{O}(mn/f)$ time, and  reduces the number of components of $\calF$ from $f$ to at most $f - f^3/(10^5 \cdot n^2)$. 
Suppose that Stage I lasts for exactly $t \geq 1$ rounds.
For each $i \in [1, t]$, let $f_i$ denote the number of components of the forest $\calF$ just after the $i^{\th}$ round.
Let $f_0 = n$.
Thus, by definition, we have $f_i \geq n^{3/4}$ for all $i \in [0, t-1]$ and $f_t < n^{3/4}$.
From the above discussion, it follows that $f_{i+1} \leq f_i - f_i^3/(10^5 \cdot n^2)$ for all $i \in [0,t-1]$.
We now derive that
$$ \frac{1}{f_i} \leq \frac{1}{f_{i+1}} \cdot \left(1 - \frac{f_i^2}{10^5n^2}\right) \leq \frac{1}{f_{i+1}} \cdot \left(1 - \frac{1}{10^5\sqrt{n}}\right),$$
where the last inequality holds because $f_{i} \geq n^{3/4}$ for all $i \in [0, t-1]$.
Hence, we get
$$\sum_{i=0}^{t-1} \frac{1}{f_i} \leq \frac{1}{f_{t-1}} \cdot \sum_{i=0}^{\infty} \left(1- \frac{1}{10^5\sqrt{n}}\right)^i \leq \frac{1}{n^{3/4}} \cdot  (10^5 \sqrt{n}) = O \!\left(n^{-1/4}\right). $$

Accordingly, we infer that the total time spent during Stage I is at most 
$\tilde{O}(\sum_{i=0}^{t-1} mn/f_i) = \tilde{O}(mn^{3/4})$. 
This gives us the desired runtime bound of \Cref{theorem:main}.

\section{Key Property of an $(H+1)$-Configuration: Proof of \Cref{key-lemma}}\label{sec:main-lemma}

Throughout this section, we make the following assumption.
\begin{assumption}
\label{assume:main}
$\left( \calF^{(s)}, \calM^{(s)}, D^{(s)}\right)$ is an $(H+1)$-configuration {\em and} $s < f^3/(10^5 \cdot n^2)$. 
\end{assumption}
Under \Cref{assume:main}, we will derive a contradiction, which in turn will imply \Cref{key-lemma}.

\subsection{Some Useful Notations and Terminologies} 
For every $t \in [0, H]$,
let $\calP_t$ denote the collection of all alternating chains of length at most $t$ w.r.t.~$\left(\calF^{(s)},\calM^{(s)}, D^{(s)}\right)$.
Consider any alternating chain $P = (w_0, z_1, w_1, \ldots, z_k, w_k) \in \calP_t$, for some $k \in [0, t]$.
Let $\tail(P)$ denote the set of all nodes $v \in V$ that satisfy the following conditions.
\begin{enumerate}
\item  $(w_k, v) \in E - E\left(\calF^{(s)}\right)$. 
\item For each $i \in [0, k]$, the node $v$ does {\em not} belong to the $i^{th}$ $\mathcal{M}^{(s)}$-block  of $P$.
\end{enumerate}

Intuitively, the set $\tail(P)$ consists of all the nodes that might potentially be appended to $P$ to form an augmenting chain. 
For each $t \in [1, H]$, we define $\tail(t) := \bigcup_{P \in \calP_{t-1}} \tail(P)$.
For consistency of notations, we define $\tail(0) := \emptyset$.
Next, for every $t \in [0, H]$, let $N(t)$ denote the collection of all nodes $u \in V$ that satisfy the following conditions.
\begin{enumerate}
\item $u$ is $\calM^{(s)}$-non-reducible.
\item There exists an alternating chain $P \in \calP_t$ of length $k$ such that $u$ belongs to the $i^{\th}$ $\calM^{(s)}$-block of $P$, for some $0 \leq i \leq k \leq t$.
\end{enumerate}

Fix any $t \in [0, H]$, and any $\calM^{(s)}$-atom $C$. We say that $C$ is a {\bf $t$-atom} iff there exists an alternating chain $P \in \calP_t$ of length $k$ such that $C$ is the $i^{th}$ critical $\calM^{(s)}$-atom of $P$, for some $0 \leq i \leq k \leq t$.
Next, for all $t \in [0,H]$, we define $W(t) := \tail(t) \cup N(t)$.
Now, consider any arbitrary molecule $M$ of $\calF^{(s)}$, and any arbitrary $X \subseteq V$.
We refer to $M$ as $X$-avoiding if $X \cap V(M) = \emptyset$ and $M$ is one of the connected components of $\calF^{(s)} - X$.
It is easy to verify that there is at most one edge $(u,v) \in E(\calF^{(s)})$ such that $u \in X$ and $v \in V(M)$ (and in that case, $u$ must be the root of $M$).
There might be other edges $(u,v) \in E - \calF^{(s)}$ where $u \in X$ and $v \in V(M)$, but there is at most one edge in $\calF^{(s)}$ that can connect $M$ to $X$.
Let $\calM(X)$ denote the collection of all $X$-avoiding molecules of $\calF^{(s)}$.
Clearly, $\calM(X)$ is a molecular decomposition of $\calF^{(s)}$.

\subsection{A Potential Function Argument}

To prove \Cref{key-lemma}, we use $\Delta^\star \cdot |W(t)|$ as a potential function. We will show that if we increase $t$ by one, then the potential increases by at least $f/5$ (see \Cref{cor:potential}). Thus, as we increase $t$ from $0$ to $H$, the potential increases by at least $fH/5 \geq 4n$, since $H = \lceil 20 n/f \rceil$. In contrast, we show that the potential is always upper bounded by $3n$ (see \Cref{lm:upper:bound:potential}), which leads to the desired contradiction.

As an intermediate step towards proving  \Cref{cor:potential}, we derive \Cref{lem:number-of-avoiding-W(t)}, which relates the growth of the potential with the number of $W(t)$-avoiding molecules of $\calF^{(s)}$.  We defer the proofs of \Cref{lm:upper:bound:potential} and \Cref{lem:number-of-avoiding-W(t)}  to \Cref{sec:proof-lem:number-of-avoiding-W(t)}, which appear {\em after} we have established some necessary properties of the set $W(t)$ in \Cref{sec:properties-tail-N-atom}.

\begin{lemma}
\label{lm:upper:bound:potential}
For every $t \in [0, H]$, we have $\Delta^\star \cdot |W(t)| \leq 3n$.
\end{lemma}

\begin{lemma}\label{lem:number-of-avoiding-W(t)}
    Consider any index $t \in [0, H-1]$. Let $\calM_{\texttt{atom}}(W(t)) \subseteq \calM(W(t))$ denote the set of $W(t)$-avoiding molecules of $\calF^{(s)}$ that are $t$-atoms. Then, the following conditions must hold.
    \begin{enumerate}    
        \item $\Delta^\star \cdot |W(t+1)| \geq$ $\left|\calM_{\texttt{atom}}(W(t))\right| + |W(t)| - 1$.

        \item  $\left|\calM(W(t))\right| \geq f/2 + (\Delta^\star-1)\cdot |W(t)|$.

        \item $\left|\calM(W(t)) - \calM_{\texttt{atom}}(W(t))\right| \leq f/4$.
    \end{enumerate}
\end{lemma}

\begin{corollary}
\label{cor:potential}
For every $t \in [0, H-1]$, we have $\Delta^\star \cdot |W(t+1)| \geq f/5 + \Delta^\star \cdot |W(t)|$.
\end{corollary}

\begin{proof}
From \Cref{lem:number-of-avoiding-W(t)}, we infer that
\begin{eqnarray*}
\Delta^\star \cdot |W(t+1)| & \geq & \left|\calM_{\texttt{atom}}(W(t))\right| + |W(t)| - 1 \\
& \geq & \left|\calM(W(t))\right| -\left( \left|\calM(W(t)) - \calM_{\texttt{atom}}(W(t))\right| \right) + |W(t)| - 1 \\
& \geq & 
f/2 + (\Delta^\star -1 ) \cdot |W(t)| - f/4 + |W(t)| - 1 \\
& \geq & f/4 - 1 + \Delta^\star \cdot |W(t)| \\
& \geq & f/5 + \Delta^\star \cdot |W(t)|.
\end{eqnarray*}
The last inequality holds since $f \geq 20$. This concludes the proof of the corollary.
\end{proof}

\paragraph{Proof of \Cref{key-lemma}.} Summing the inequality from \Cref{cor:potential} over all $t \in [0, H-1]$, we get
\begin{equation}
\label{eq:contradict}
\Delta^\star \cdot |W(H)| \geq fH/5 + \Delta^\star \cdot |W(0)| \geq fH/5.
\end{equation}

In contrast, \Cref{lem:number-of-avoiding-W(t)} guarantees that $3n \geq \Delta^\star \cdot |W(H)|$. Combining this with \Cref{eq:contradict}, we get $3 n \geq f H/5$. However, this leads to a contradiction, since  $H = \lceil 20 n/f \rceil$ as per the statement of \Cref{key-lemma}. Thus, it must be the case that \Cref{assume:main} does {\em not} hold. This implies \Cref{key-lemma}.

\subsection{Basic Properties of the Set $W(t)$}\label{sec:properties-tail-N-atom}

In this section, we derive some basic properties of the nodes in $W(t)$, which would subsequently be used in \Cref{sec:proof-lem:number-of-avoiding-W(t)} while proving  \Cref{lm:upper:bound:potential} and \Cref{lem:number-of-avoiding-W(t)}.

\begin{observation}\label{lem:increasing-tail(t)-N(t)}
    For all $t \in [0, H-1]$, we have $\tail(t) \subseteq \tail(t+1)$ and $ N(t) \subseteq N(t+1)$.
\end{observation}

\begin{proof}
Follows immediately from the definitions of the sets $\tail(t)$ and $N(t)$.
\end{proof}

\begin{lemma}\label{lem:outgoing-edges-of-t-atoms}
    Fix any $t \in [0, H-1]$, and consider any edge $(u,v) \in E$ such that $u \in V(C)$ and $v \notin V(C)$ for some $t$-atom $C$.
    Then, we must have $v \in \tail(t+1) \cup N(t)$.
\end{lemma}

\begin{proof}
W.l.o.g., let $C$ be the $k^{th}$ critical $\calM^{(s)}$-atom of an alternating chain $P \in \calP_t$ of length $k$, for some $k \in [0, t]$, and let $P = (w_0, z_1, w_1, \ldots, z_k, w_k)$, with $w_k = u$.
    
    We now consider the following mutually exclusive and exhaustive cases.

    \medskip
    \noindent 
    \textbf{Case 1: $(u,v) \in E(\calF^{(s)})$.} Since $C$ is an  $\calM^{(s)}$-atom, $u \in V(C)$ and $v \notin V(C)$, either $C$ itself is a normal $\calM^{(s)}$-molecule with $v$ being its  root, or $v$ is an $\calM^{(s)}$-non-reducible node inside the $\calM^{(s)}$-molecule containing $C$. Thus, either $v = z_k$, or $v$ is an $\calM^{(s)}$-non-reducible node on the path $P_{w_k, z_k}^{\calF}$. In the former event we have $v \in \tail(k+1) \subseteq \tail(t+1)$ as per \Cref{lem:increasing-tail(t)-N(t)}, whereas in the latter event we have $v \in N(t)$.

    \medskip
    \noindent 
    \textbf{Case 2. $(u, v) \notin E(\calF^{(s)})$ and the node $v$  belongs to  the $k^{th}$ $\calM^{(s)}$-block of $P$.}
    Since the node $u$ also belongs to the $k^{th}$ $\calM^{(s)}$-block of $P$ and is $\calM^{(s)}$-reducible, the node $v$ must be $\calM^{(s)}$-non-reducible, for otherwise, we would have $v \in V(C)$ according to the definition of $\calM^{(s)}$-atoms. 
    It follows that $v \in N(t)$.

    \medskip
    \noindent 
    \textbf{Case 3. $(u, v) \notin E(\calF^{(s)})$ and $v$ appears in the $i^{th}$ $\calM^{(s)}$-block of $P$, for some $i \in [0, k-1]$.}
    If $v$ is $\calM^{(s)}$-non-reducible, we have $v \in N(t)$ as desired.
    Otherwise, $v$ is contained in an $\calM^{(s)}$-atom.
    Define the following sequence $ A := (w_0,z_1, w_1, \ldots, z_i, v, u)$.
    It is straightforward to verify that $A$ is an augmenting chain of length $i+1 \leq k \leq H-1$ w.r.t.~$\left(\calF^{(s)}, \calM^{(s)}, D^{(s)}\right)$.
    Since $\left(\calF^{(s)}, \calM^{(s)}, D^{(s)}\right)$ is an $(H+1)$-configuration, this is a contradiction.

    \medskip
    \noindent 
    \textbf{Case 4. $(u, v) \notin E(\calF^{(s)})$, and $v$ does {\em not} belong to the $i^{th}$ $\calM^{(s)}$-block of $P$, for each $i \in [0, k]$.}
    In this case, we must necessarily have $v \in \tail(k+1)$ by definition. Since $k \in [0, t]$, we have $\tail(k+1) \subseteq \tail(t+1)$ as per \Cref{lem:increasing-tail(t)-N(t)}. Thus, we infer that $v \in \tail(t+1)$.
\end{proof}

\begin{corollary}\label{cor:outgoing-edges-of-t-atoms}
    Fix any $t \in [0, H-1]$, and consider any edge $(u,v) \in E$ such that $u \in V(C)$ and $v \notin V(C)$ for some $t$-atom $C$.
    Then, we must have $v \in W(t+1)$.
\end{corollary}

\begin{proof}
As $W(t+1) = \tail(t+1) \cup N(t+1)$, the corollary follows from \Cref{lem:increasing-tail(t)-N(t)} and \Cref{lem:outgoing-edges-of-t-atoms}.
\end{proof}

\begin{lemma}\label{Wt-has-full-degree}
    For every $t \in [1, H]$ and every node $x \in W(t) - D^{(s)}$, we have $\deg_{\calF^{(s)}}(x) = \Delta^\star+1$.
\end{lemma}

\begin{proof}
Fix any $t \in [1, H]$ and any node $x \in W(t) - D^{(s)}$.
If $x \in N(t)$, then $x$ is $\calM^{(s)}$-non-reducible, and hence $\deg_{\calF^{(s)}}(x) = \Delta^\star+1$. 
For the rest of the proof, suppose that $x \in \tail(t) - \left(D^{(s)} \cup N(t) \right)$. 

Accordingly, there is an alternating chain $P = (w_0, z_1, w_1, \ldots, z_k, w_k) \in \calP_{t-1}$ of length $k \in [0, t-1]$ such that $(w_k, x) \in E - E\left(\calF^{(s)}\right)$, and $x$ lies outside the $i^{th}$ $\calM^{(s)}$-block of $P$ for all $i \in [0, k]$. Since $x \notin \left(D^{(s)} \cup N(t)\right)$, we are now left with the following possibilities.

\medskip
\noindent
{\bf Case 1: $\deg_{\calF^{(s)}}(x) \leq \Delta^\star$, and $x$ is either $\calM^{(s)}$-reducible or $\calM^{(s)}$-free.} Here, the sequence $(w_0, z_1, w_1, \ldots, z_k, w_k, x)$, obtained by appending $x$ to $P$, is an augmenting chain of length $k+1 \leq t \leq H$ w.r.t.~$\left( \calF^{(s)}, \calM^{(s)}, D^{(s)}\right)$. This contradicts \Cref{assume:main}. So, we never end up in this case.

\medskip
\noindent
{\bf Case 2: $\deg_{\calF^{(s)}}(x) = \Delta^\star + 1$.} Here, we already have the desired guarantee on the degree of $x$.
\end{proof}

\begin{corollary}\label{sum-of-deg-Wt}
    For every $t \in [1, H]$, we have
    $\sum_{x \in W(t)} \deg_{\calF^{(s)}}(x) \geq (\Delta^\star + 1) \cdot |W(t)|  - f/4$.
\end{corollary}

\begin{proof}
 Observe that  $\left| D^{(i)} - D^{(i-1)}\right| \leq H$ for each $i \in [1, s]$, since each augmenting chain $A^{(i)}$ is of length at most $H$ (see the statement of \Cref{key-lemma}). Moreover, we have $D^{(0)} = \emptyset$. This implies that
\begin{equation}
\label{eq:new}
\left| D^{(s)} \right| \leq \left| D^{(0)} \right| + \sum_{i=1}^{s} \left| D^{(i)} - D^{(i-1)} \right| \leq s \cdot H.
\end{equation}

By definition, every node $x \in D^{(s)}$ satisfies $\deg_{\calF^{(s)}}(x) = \Delta^\star$. This observation, along with \Cref{eq:new} and \Cref{Wt-has-full-degree}, gives us:
\begin{eqnarray*}
\sum_{x \in W(t)} \deg_{\calF^{(s)}}(x) & = &  (\Delta^\star+1) \cdot |W(t)| - \left| W(t) \cap D^{(s)} \right| \\
 & \geq & (\Delta^\star+1) \cdot |W(t)|   - \left| D^{(s)} \right| \\
 & \geq & (\Delta^\star+1) \cdot |W(t)|  - s \cdot H \\
 & \geq & (\Delta^\star+1) \cdot |W(t)|  - f/4.
\end{eqnarray*}
The last inequality in the above derivation follows from \Cref{assume:main}, and the fact that $H = \lceil 20 n/f \rceil$ as per the statement of \Cref{key-lemma}. This concludes the proof of the corollary.
\end{proof}

\begin{lemma}\label{Wt-not-in-atom}
    For every $t \in [1, H]$ and $x \in W(t)$, the node $x$ is {\em not} part of  {\em any} $\calM^{(s)}$-atom.
\end{lemma}

\begin{proof}
The argument is very similar to the proof of \Cref{Wt-has-full-degree}. Consider any $t \in [1, H]$ and any node $x \in W(t) = \tail(t) \cup N(t)$. If $x \in N(t)$, then by definition $x$ is $\calM^{(s)}$-non-reducible, and so $x$ can {\em not} be part of any $\calM^{(s)}$-atom. For the rest of the proof, suppose that $x \in \tail(t) - N(t)$.

Accordingly, there is an alternating chain $P = (w_0, z_1, w_1, \ldots, z_k, w_k) \in \calP_{t-1}$ of length $k \in [0, t-1]$ such that $(w_k, x) \in E - E\left(\calF^{(s)}\right)$, and $x$ lies outside the $i^{th}$ $\calM^{(s)}$-block of $P$ for all $i \in [0, k]$. 

Now, if $x$ is part of some $\calM^{(s)}$-atom, then the sequence $(w_0, z_1, w_1, \ldots, z_k, w_k, x)$, obtained by appending $x$ to $P$, is an augmenting chain of length $k+1\leq t 
\leq H$ w.r.t.~$\left( \calF^{(s)}, \calM^{(s)}, D^{(s)}\right)$. This contradicts \Cref{assume:main}. Thus, we must have that $x$ is {\em not} part of any $\calM^{(s)}$-atom.
\end{proof}

\subsection{Analyzing the Potential Function: Proofs of \Cref{lm:upper:bound:potential} and \Cref{lem:number-of-avoiding-W(t)}}\label{sec:proof-lem:number-of-avoiding-W(t)}

\subsubsection{Proof of \Cref{lm:upper:bound:potential}}

As the sum of the degrees of all the nodes in any forest is at most $2n$, from \Cref{sum-of-deg-Wt} we infer that

\begin{align*}
   2n &\geq \sum_{x \in W(t)} \deg_{\calF^{(s^\star)}}(x) \geq (\Delta^\star + 1) \cdot |W(t)| - f/4 \geq \Delta^\star \cdot |W(t)| - n 
\end{align*}
Rearranging the terms, we get $\Delta^\star \cdot |W(t)| \leq 3n$.

\subsubsection{Proof of the First Property of \Cref{lem:number-of-avoiding-W(t)}}

Fix any $t \in [0, H-1]$.
Define $\calC_t$ to be the collection of all node sets $C$ satisfying one of the following;
\begin{itemize}
    \item Either $C$ is the node set of a $W(t)$-avoiding molecule of $\calF^{(s)}$ which is a $t$-atom as well, i.e., $C = V(M)$ for some $M \in \calM_{\texttt{atom}}(W(t))$,
    \item Or $C$ is a singleton set consisting of a node $x \in W(t) = \tail(t) \cup N(t)$.
\end{itemize}
 According to \Cref{Wt-not-in-atom}, the collection $\calC_t$ is well-defined since no nodes in $W(t)$ is part of any $\calM^{(s)}$-atom, every component in  $\calM_{\texttt{atom}}(W(t))$ is a $t$-atom, and each $t$-atom is an $\calM^{(s)}$-atom.
The optimal spanning tree $T^\star$, with maximum degree $\Delta^\star$, connects the node sets of $\calC_t$ together.
Thus, there must exist at least $\left| \calC_t \right| - 1$ edges $(u, v) \in E(T^\star)$ with one endpoint $u$ being part of some node set $C_u \in \calC_t$ and the other endpoint $v$ being outside that node set $C_u$.
Since $W(t) \subseteq W(t+1)$, \Cref{cor:outgoing-edges-of-t-atoms}
implies that each of these edges must be incident on $W(t+1)$.
Accordingly, we get
\begin{align*}
   \Delta^\star \cdot |W(t+1)| \geq \sum_{x \in W(t+1)} \deg_{T^\star}(x) \geq |\calC_t|-1 = \left| \calM_{\texttt{atom}}(W(t)) \right| + |W(t)| - 1.
\end{align*}

\subsubsection{Proof of the Second property of \Cref{lem:number-of-avoiding-W(t)}}

Fix any $t \in [0, H-1]$. We start with the following important claim.

\begin{claim}
    \label{claim:num-incident-on-S}
    Consider any component $C$ of $\calF^{(s)}$, and let $\calM_C(W(t)) \subseteq \calM(W(t))$ denote the collection of $W(t)$-avoiding molecules that are subgraphs of $C$. Then, we must have
   $$\left| \calM_C(W(t)) \right| \geq  1 + \sum_{x\in W(t) \cap V(C)}(\deg_{\calF^{(s)}}(x) - 2).$$ 
\end{claim}

\begin{proof}
The claim is trivial if $W(t) \cap V(C) = \emptyset$.
For the rest of the proof, suppose that $W(t) \cap V(C) \neq \emptyset$.

We select {\em any} node $v \in V(C)$, and treat $C$ as a  tree rooted at $v$.
Next, we sort the nodes in $W(t) \cap C$ in non-decreasing order of their depths in this rooted tree $C$.
Suppose that $W(t) \cap V(C):=\{x_1,x_2,\ldots,x_\ell\}$ in  this sorted order. 
Thus, for all $1 \leq j \leq i \leq \ell$, the depth of $x_i$ is at least the depth of $x_j$ in $C$. For each $i \in [1, \ell]$, let $\calM_i$ denote the collection of all $\{x_1, \ldots, x_i\}$-avoiding molecules that are subgraphs of $C$. It is easy to observe that $\calM_{\ell} = \calM_C(W(t))$.
For consistency of notations, we define $\calM_0 = C$.
Below, we track the growth of $\left| \calM_i\right|$ as a function of $i \in [0, \ell]$. Intuitively, this can be visualized by considering a process whereby we remove the nodes  $x_1,x_2,\ldots,x_\ell$ (along with their incident edges) one after another from $C$, and we keep track of the number of components created during this process that also happen to be molecules of $\calF^{(s)}$.

In the beginning, we have 
\begin{equation}
\label{eq:count:0}
\left|\calM_0\right| = 1.
\end{equation}
 Once the node $x_1$ gets deleted, we have $\deg_{\calF^{(s)}}(x_1)$ components created during the process that are molecules in  $\calF^{(s)}$, and so $\left|\calM_1\right| = \deg_{\calF^{(s)}}(x_1)$. Thus, we infer that
\begin{equation}
\label{eq:count:1}
\left|\calM_1\right| -  \left|\calM_0\right|= \deg_{\calF^{(s)}}(x_1) - 1 \geq \deg_{\calF^{(s)}}(x_1) - 2.
\end{equation}
 Now, consider any $i \in [2, \ell]$. Just before removing the node $x_i$, there are at least $\deg_{\calF^{(s)}}(x_i)-1$ edges incident on it (one incident edge between $x_i$ and its parent in the rooted tree $C$ might have been deleted previously during the process, since we are considering the nodes $x_1, x_2, \ldots$ in non-decreasing order of their depths in $C$). Accordingly, as we remove $x_i$, the number of components created during the process that also happen to be molecules in $\calF^{(s)}$ increases by at least $\deg_{\calF^{(s)}}(x_i)-1$ (the sub-trees of $C$ rooted at children of $x_i$), and one component is damaged since we removed $x_i$.
 Hence, we get
 \begin{equation}
\label{eq:count:2}
\left|\calM_i\right| -  \left|\calM_{i-1}\right|= \deg_{\calF^{(s)}}(x_i) - 2 \text{ for all } i \in [2, \ell].
\end{equation}
Since $\calM_{\ell} = \calM_C(W(t))$ and $W(t) \cap V(C) = \{x_1, \ldots, x_\ell\}$, the lemma follows from summing up \Cref{eq:count:0}, \Cref{eq:count:1} and \Cref{eq:count:2} for all $i \in [2, \ell]$.
\end{proof}

Observe that each time we apply an augmenting chain w.r.t.~$(\calF^{(i)}, \calM^{(i)}, D^{(i)})$ for each $i \in [0, s-1]$, the number of components in the underlying valid forest decreases by one. Since there are $f$ components in $\calF^{(0)}$, it follows that there are $f - s \geq 3f/4$ components in $\calF^{(s)}$ (see the statement of \Cref{key-lemma}), where the inequality follows from \Cref{assume:main}. 

Now, summing the inequalities guaranteed by \Cref{claim:num-incident-on-S} over all the components of $\calF^{(s)}$, we get
\begin{eqnarray*}
\left| \calM(W(t)) \right| \geq  3f/4 + \sum_{x\in W(t)}(\deg_{\calF^{(s)}}(x) - 2) \geq f/2 + (\Delta^\star-1) \cdot |W(t)| ,
\end{eqnarray*}
where the last inequality follows from \Cref{sum-of-deg-Wt}. This concludes the proof.

\subsubsection{Proof of the Third Property of \Cref{lem:number-of-avoiding-W(t)}}

Fix any $t \in [0, H-1]$. For each $i \in [0, s]$, let $\free(i) \subseteq V$ denote the set of all $\calM^{(i)}$-free nodes (recall $\calM^{(i)}$ from \Cref{key-lemma}). We start with an observation which tracks the growth of these sets $\free(i)$ over $i \in [0, s]$.

\begin{observation}
\label{ob:nested}
For each $i \in [1, s]$, we have $\free(i-1) \subseteq \free(i)$, and $\left| \free(i) - \free(i-1)\right| \leq (2\theta) \cdot (H+1)$.
\end{observation}

\begin{proof}
After applying the augmenting chain $A^{(i)}$ w.r.t.~$\left(\calF^{(i-1)}, \calM^{(i-1)}, D^{(i-1)}\right)$, we obtain $\calM^{(i)}$ by removing the {\em affected} molecules from $\calM^{(i-1)}$ (see the discussion  after \Cref{def:aug-chain}). The set $\free(i)$ consists of all the nodes in $\free(i-1)$, {\em plus} all the nodes in the affected $\calM^{(i-1)}$-molecules.  

Finally, we note that there can be at most $H+1$ affected molecules in $\calM^{(i-1)}$ due to the application of the augmenting chain $A^{(i)}$ (since $A^{(i)}$ is of length at most $H$), and each affected molecule contains at most $2\theta$ nodes (see properties \ref{property:theta-max-decomposition-1} and \ref{property:theta-max-decomposition-2} of the initial $\theta$-molecular decomposition $\calM^{(0)}$ of $\calF^{(0)}$).
\end{proof}

\begin{corollary}
\label{cor:nested}
We have $\left| \free(s) - \free(0)\right| \leq s \cdot (2\theta) \cdot (H+1)$.
\end{corollary}

\begin{proof}
Follows immediately from \Cref{ob:nested}.
\end{proof}

Next, we show that a $W(t)$-avoiding molecule that is {\em not} a $t$-atom must contain an $M^{(s)}$-free node.

\begin{claim}
\label{cl:free}
   Consider any $M \in \calM(W(t)) - \calM_{\texttt{atom}}(W(t))$.
    Then, we must have $V(M) \cap \free(s) \neq \emptyset$. 
\end{claim}

\begin{proof}
For the sake of contradiction, assume that $V(M) \cap \free(s) = \emptyset$.
According to the properties of the initial $\theta$-molecular decomposition $\calM^{(0)}$ and the fact that $\calM^{(s)} \subseteq \calM^{(0)}$, we conclude that $M$ is contained entirely in an $\calM^{(s)}$-molecule (say $T \in \calM^{(s)}$).
The reason is that the normal $\calM^{(0)}$-molecules (resp.~normal $\calM^{(s)}$-molecules) are separated by their roots that are $\calM^{(0)}$-free nodes (resp.~$\calM^{(s)}$-free nodes) according to the definition of a molecular decomposition.
Hence, we have the following two cases;

\medskip
\noindent
\textbf{Case I.}
$T$ is a special $\calM^{(s)}$-molecule.
In this case, all of the $\calM^{(s)}$-non-reducible nodes in $T$ are part of $N(0) \subseteq W(t)$.
Since $M$ is $W(t)$-avoiding and $W(t)$ does not contain any $\calM^{(s)}$-reducible node (see \Cref{Wt-not-in-atom}), we conclude that $M$ itself must be an $\calM^{(s)}$-atom inside $T$ which is a $0$-atom (as well as a $t$-atom) by definition.
But, this is in contradiction with the assumption that $M \notin \calM_{\texttt{atom}}(W(t))$.

\medskip
\noindent
\textbf{Case II.}
$T$ is a normal $\calM^{(s)}$-molecule.
According to the definition of a $W(t)$-avoiding molecule, since $M \subseteq T$ is not a special molecule of $\calF^{(s)}$, there must exist an edge $(x,y) \in \calF^{(s)}$ with $y \in W(t)$ such that $M = T^{\calF^{(s)}}_{xy}$.
Since $W(t)= N(t) \cup \tail(t)$, we have two sub-cases as follows.

\begin{itemize}
    \item $y \in \tail(t)$.
    Then, there exists an alternating chain $P = (w_0,z_1,w_1,\ldots,z_k,w_k) \in \calP_{t-1}$ of length $k \in[0,t-1]$ such that $y \in \tail(P)$.
    Now, according to the definition of $N(t)$, any $\calM^{(s)}$-non-reducible node in $T^{\calF^{(s)}}_{xy}$ is contained in $N(t)$.
    But $V(M) \cap N(t) = \emptyset$ since $M = T^{\calF^{(s)}}_{xy}$ is $W(t)$-avoiding.
    We conclude that $V(M)$ does not contain any $\calM^{(s)}$-non-reducible node and must be an $\calM^{(s)}$-atom.
    Finally, $(w_0,z_1,w_1,\ldots,z_k,w_k, y, x) \in \calP_t$ is an alternating chain of length $k+1 \leq t$ such that $M$ is its $(k+1)^\th$ critical $\calM^{(s)}$-atom.
    Hence, $M$ must be a $t$-atom, which is in contradiction with the assumption $M \notin \calM_{\texttt{atom}}(W(t))$.

    \item $y \in N(t)$.
    Then, there exists an alternating chain $P = (w_0,z_1,w_1,\ldots,z_k,w_k) \in \calP_t$ such that $y$ belongs to the $i^\th$ $\calM^{(s)}$-block of $P$ for some $0 \leq i \leq k \leq t$.
    This alternating chain $P$ certifies that all of the $\calM^{(s)}$-non-reducible nodes of $M = T^{\calF^{(s)}}_{xy}$ belongs to $N(t) \subseteq W(t)$ as well.
    Since $M$ is $W(t)$-avoiding, we conclude that it is an $\calM^{(s)}$-atom.
    Finally, $P' = (w_0,z_1,w_1,\ldots,z_i,x) \in \calP_t$ is an alternating chain where $M$ is its $i^\th$ critical $\calM^{(s)}$-atom.
    This concludes $M$ is a $t$-atom which is in contradiction with $M \notin \calM_{\texttt{atom}}(W(t))$.\qedhere
\end{itemize}
\end{proof}

Let $\calM^\star(W(t)) \subseteq \calM(W(t))$ denote the collection of $W(t)$-avoiding molecules that  contain at least one node from the set $\free(s)$.  \Cref{cl:free} implies that
\begin{equation}
\label{eq:partition}
\calM(W(t)) - \calM_{\texttt{atom}}(W(t)) \subseteq \calM^\star(W(t)).
\end{equation}
Accordingly, for the rest of the proof, we focus on upper bounding $\calM^\star(W(t))$. Towards this end, we partition $\calM^\star(W(t))$ into two subsets -  $\calM^\star_1(W(t))$ and $\calM^\star_2(W(t))$ - as defined below, and separately upper bound the size of each of these two subsets.

\medskip
\noindent {\bf I: The set $\calM^\star_1(W(t))$.} This set consists of the collection of $W(t)$-avoiding molecules in $\calM^\star(W(t))$ that contain at least one node from $\free(s) - \free(0)$. Since the $W(t)$-avoiding molecules in $\calM^\star_1(W(t))$ are mutually node-disjoint, \Cref{cor:nested} implies that 
\begin{equation}
\label{eq:partition:1}
\left| \calM^\star_1(W(t)) \right| \leq \left| \free(s) - \free(0)\right| \leq s \cdot (2\theta) \cdot (H+1). 
\end{equation}

\medskip
\noindent {\bf II: The set $\calM^\star_2(W(t))$.} This set consists of all the $W(t)$-avoiding molecules  $M \in \calM^\star(W(t))$ such that every $\free(s)$ node in $M$ is also part of $\free(0)$, i.e., $V(M) \cap \free(s) = V(M) \cap \free(0) \neq \emptyset$. In other words, $M$ does not contain {\em any} node which was part of a molecule that got affected while applying the sequence of augmenting chains $A^{(1)}, \ldots, A^{(s)}$. Hence, $M$ is also a molecule of $\calF^{(0)}$, but is not  an $\calM^{(0)}$-molecule since it contains an $\calM^{(0)}$-free node (part of $\free(0)$).
Since the initial molecular decomposition $\calM^{(0)}$ is a $\theta$-molecular decomposition of $\calF^{(0)}$, it follows that $|V(M)| \geq \theta+1$ (see property \ref{property:theta-max-decomposition-3}).
As the molecules in $\calM^\star_2(W(t))$ are mutually node-disjoint, we get
\begin{equation}
\label{eq:partition:2}
\left| \calM^\star_2(W(t)) \right| \leq n/(\theta+1).
\end{equation}

Since $\theta = H = \lceil 20 n/f \rceil$ as per the statement of \Cref{key-lemma} and $s < f^3/(10^5 \cdot n^2)$ as per \Cref{assume:main}, from \Cref{eq:partition}, \Cref{eq:partition:1} and \Cref{eq:partition:2}, we infer that
$$\left|\calM(W(t)) - \calM_{\texttt{atom}}(W(t)) \right| \leq s \cdot (2\theta) \cdot (H+1) + n/(\theta+1) \leq f/4.$$
This concludes the proof.

\section{Computing an $(\ell+1)$-Configuration: Proof of \Cref{lem:subroutine-main}}\label{sec:description-subroutine-main}

In this section, we provide the most critical subroutines of our algorithm for \Cref{lem:subroutine-main}.
Before we explain how the subroutine works, we provide some properties about augmenting chains in \Cref{sec:algorithmic-insights} that help us design the algorithm and provide a better understanding of the algorithm.
Then, we describe the algorithm for \Cref{lem:subroutine-main} in \Cref{sec:subroutine-main-description}.
Finally, we analyze the algorithm in \Cref{sec:subroutine-main-analysis} that proves \Cref{lem:subroutine-main}.
Throughout this section, we consider the following assumption.

\begin{assumption}\label{assumption:input}
    An $\ell$-configuration $(\calF, \calM, D)$ is given for some $\ell \geq 1$.
    The goal is to apply as many augmenting chains of length $\ell$ as possible to $(\calF, \calM, D)$ and achieve an $(\ell + 1)$-configuration $(\calF', \calM', D')$.
\end{assumption}

\noindent
\textbf{Notation.}
For simplicity in presenting the algorithm, for each special $\calM$-molecule $M$, we consider a \textbf{dummy} node $r_M$ as the root of this molecule and attach it to an arbitrary node $u \in V(M)$ with a dummy edge $(r_M, u)$.
Then, we treat $M$ as a rooted tree with $r_M$ as its root.
Note that these dummy roots are not contained in $V$, the dummy edges are not contained in $E$ and they do not contribute to the degree of nodes in $\calF$.
Moreover, by considering these dummy roots, we can denote the $0^\th$ $\calM$-block of an augmenting chain $(w_0,z_1,w_1,\ldots,z_k,w_k)$ w.r.t.~$(\calF, \calM, D)$ by $T^{\calF}_{w_0 \leftarrow z_0}$, where $z_0$ is the dummy root of the special $\calM$-molecule containing $w_0$. 
This simplifies the further arguments in this section.
Throughout this section, we \emph{implicitly} consider this notation $z_0$ according to the specified augmenting chain.

For each $x \in V$, we define the sub-tree $T^\calF_x$ of $\calF$ as follows;
\begin{itemize}
    \item If $x$ is the root of an $\calM$-molecule,
    then $T^\calF_x$ is defined as the sub-tree of $\calF$ consisting of all normal $\calM$-molecules rooted at $x$.

    \item If $x$ is inside an $\calM$-molecule $M$,
    then $M$ is treated as a rooted sub-tree of $\calF$, i.e., if $M$ is a normal $\calM$-molecule, we have already defined its root in previous sections, and if $M$ is a special $\calM$-molecule, we defined a dummy root for $M$ in the previous paragraph.
    Then, $T^\calF_x$ is defined as the sub-tree rooted at $x$ while considering $M$ as a rooted sub-tree of $\calF$.

    \item If $x$ is $\calM$-free,
    then we define $T^\calF_x =\{x\}$.
\end{itemize}
Moreover, for every dummy root $x$ of an special $\calM$-molecule, we define $T^{\calF}_x$ as the entire $\calM$-molecule containing $x$.

\subsection{Useful Properties of Minimum Length Augmenting Chains}\label{sec:algorithmic-insights}

We provide some properties of augmenting chains that help us find them algorithmically.

\subsubsection{Pseudo-Augmenting Chains}
We start by defining the notion of a pseudo-augmenting chain as follows.

\begin{definition}[Pseudo-Augmenting Chain]\label{def:pseudo-chain}
    Assume $S = (w_0, z_1, w_1, \ldots, z_k ,w_k, z_{k+1}) $ is a sequence of nodes (not necessarily distinct) for some $k \geq 0$.
    We call $S$ a pseudo-augmenting chain of length $k+1$ w.r.t.~a valid configuration $(\calF, \calM, D)$, if it satisfies the following conditions.
    \begin{enumerate}
        \item $P(S) := (w_0, z_1, w_1, \ldots, z_k, w_k)$ satisfies properties \ref{property-z0}, \ref{property-wi-zi-connection}, and \ref{property-non-forest-edge} in \Cref{def:alt-chain} w.r.t.~$(\calF, \calM, D)$.
        
        \item $S$ satisfies properties \ref{property-well-zell1} and \ref{property-zell1-reducible} in \Cref{def:aug-chain} w.r.t.~$(\calF, \calM, D)$.
        
        \item\label{property3:pseudo-chain} $w_k$ and $z_{k+1}$ are not in the same $\calM$-atom.
    \end{enumerate}
\end{definition}

\begin{claim}\label{claim:pseudoChain-into-chain}
    Assume $(\calF, \calM, D)$ is a valid configuration.
    If there exists a pseudo-augmenting chain $S = (w_0,z_1,w_1,\ldots, z_k, w_k, z_{k+1})$ of length $k+1$ w.r.t.~$(\calF, \calM, D)$, then there exists an augmenting chain of length at most $k+1$ w.r.t.~$(\calF, \calM, D)$.
    Moreover, if $S$ itself is not an augmenting chain, then there exists an augmenting chain of length at most $k$ w.r.t.~$(\calF, \calM, D)$.
\end{claim}

\begin{proof}
We prove the claim by induction on the length of $S$.
If the length of $S$ is one, property \ref{property-wi-outside-jth-block} in \Cref{def:alt-chain} becomes vacuously true for $P(S) = (w_0)$ w.r.t.~$(\calF, \calM, D)$.
Property \ref{property-zell1-outside-blocks} in \Cref{def:aug-chain} and the fact that $w_0$ and $z_1$ are distinct follow from the assumption that $z_1$ and $w_0$ are not in the same $\calM$-atom (see property \ref{property3:pseudo-chain} in \Cref{def:pseudo-chain}).
Now, we show the induction step.
We show that if $S$ itself is not an augmenting chain, we can construct another pseudo-augmenting chain $S'$ of strictly smaller length.
This completes the proof of the both parts of the claim.
If $S$ is not an augmenting chain, we have one of the following cases.

\medskip
\noindent
\textbf{Case I: $w_0,z_1,w_1,\ldots, z_k,w_k,z_{k+1}$ Are Not Distinct.}
Since $w_0, \ldots, w_k$ are contained in $\calM$-atoms, $z_1, \ldots, z_k$ are not contained in $\calM$-atoms, and $z_{k+1}$ might or might not be contained in an $\calM$-atom, we have one of the following sub-cases.
\begin{itemize}
    \item $w_j = w_i$ for some $0\leq j < i \leq k$. In this case, it is straightforward to see that
    $$ S' := (w_0,z_1,w_1,\ldots,z_j,w_j,z_{i+1},w_{i+1},\ldots, z_k,w_k,z_{k+1})$$ is a pseudo-augmenting chain of length strictly less than $k+1$.

    \item $z_j = z_i$ for some $1 \leq j < i \leq k+1$. Similar to the previous case, it is straightforward to see that
    $$ S' := (w_0,z_1,w_1,\ldots,z_{j-1},w_{j-1},z_{i},w_{i},\ldots, z_k,w_k,z_{k+1}) $$ is a pseudo-augmenting chain of length strictly less than $k+1$.

    \item $z_{k+1} = w_i$ for some $i \in [0, k]$.
    It is obvious that $i \neq k$ since $w_k$ and $z_{k+1}$ are not in the same $\calM$-atom (see property \ref{property3:pseudo-chain} in \Cref{def:pseudo-chain}).
    Now, define
    $$ S' := (w_0,z_1,w_1,\ldots,z_{i-1},w_{i-1},z_{i},z_{k+1},w_k).$$
    Notice the change of the order of $w_k$ and $z_{k+1}$ at the end of the sequence.
    It is straightforward to see that $S'$ is a pseudo-augmenting chain of length strictly less than $k+1$.
\end{itemize}

\medskip
\noindent
\textbf{Case II: Property \ref{property-zell1-outside-blocks} in \Cref{def:aug-chain} is Violated.}
Assume that $z_{k+1}$ is in the $j^\th$ $\calM$-block of $S$ for some $j \in [0, k]$, i.e., $z_{k+1} \in T_{w_j \leftarrow z_j}^\calF$.
First, we show that $j \neq k$.
Assume the contrary, that concludes both of $w_k$ and $z_{k+1}$ must be inside the same $\calM$-molecule since both of them are in $\calM$-blocks rooted at $z_k$.
According to property \ref{property-well-zell1} in \Cref{def:aug-chain}, there is a non-forest edge between $w_k$ and $z_{k+1}$.
But, in this case, the $\calM$-atoms containing $w_k$ and $z_{k+1}$ must be merged together according to the process of defining $\calM$-atoms in each $\calM$-molecule.
This is in contradiction with the assumption that $S$ is a pseudo-augmenting chain (see property \ref{property3:pseudo-chain} in \Cref{def:pseudo-chain}).
Hence, we have $j < k$.
Now, define
$$ S' := (w_0,z_1,w_1,\ldots,z_{j-1},w_{j-1},z_{j},z_{k+1},w_k).$$
Notice the change of the order of $w_k$ and $z_{k+1}$ at the end of the sequence.
Since $j < k$, it is straightforward to see that $S'$ is a pseudo-augmenting chain of length strictly less than $k+1$.

\medskip
\noindent
\textbf{Case III: Property \ref{property-wi-outside-jth-block} in \Cref{def:alt-chain} is Violated.}
There exists $0 \leq j < i \leq k$ such that the $\calM$-atom containing $w_i$ is not outside of $T^{\calF}_{w_j \leftarrow z_j}$.
According to \Cref{obs:block}, we conclude that the $i^\th$ critical $\calM$-atom of $S$ (containing $w_i$) must be completely inside $T^\calF_{w_j \leftarrow z_j}$.
We conclude that $T^\calF_{w_j \leftarrow z_j} = T^\calF_{w_i \leftarrow z_j}$.
Now, define 
$$ S' := (w_0,z_1,w_1,\ldots,z_{j},w_{i},z_{i+1},w_{i+1}, \ldots, z_k, w_k, z_{k+1}). $$
It is straightforward to see that $S'$ is a pseudo-augmenting chain of length strictly less than $k+1$.
\end{proof}

\Cref{claim:pseudoChain-into-chain} implies that if we look for augmenting chains of \textbf{minimum} length, we can simply ignore property \ref{property-wi-outside-jth-block} in \Cref{def:alt-chain} and property \ref{property-zell1-outside-blocks} in \Cref{def:aug-chain}, i.e., instead of searching for augmenting-chains, we can search for pseudo-augmenting chains that have less conditions in the definition.
This makes the search procedure easier since we need to satisfy fewer properties.

\subsubsection{A Useful Layering}\label{sec:otpimal-layers}

Assume that $\calP^\star$ is the set of all augmenting chains of length $\ell$ w.r.t.~$(\calF, \calM, D)$.
For every $t \in [1, \ell-1]$ define
\begin{align*}
   Z^\star_t &:= \left\{ x \in V \mid \exists \ (w_0,z_1,w_1,\ldots,z_{\ell-1},w_{\ell-1}, z_{\ell}) \in \calP^\star \text{ s.t. } x = z_{t} \right\}.
\end{align*}
For the special case $t=0$, we define $Z_0^\star$ as the set of all dummy roots of special $\calM$-molecules.
The main property about these sets $Z_t^\star$ that we show in \Cref{claim:property-optimal-layers} is that they are mutually disjoint.
This property facilitates the algorithmic search for augmenting chains $A = \calP^\star$ in the following manner; the sets of potential nodes $x$ appearing as the root of the $t^\th$ $\calM$-block of $A$ for different values of $t \in [0, \ell-1]$ are mutually disjoint.
The main reason for this property is that the minimum length of any augmenting chain w.r.t.~$(\calF, \calM, D)$ is at least $\ell$ (see \Cref{assumption:input}).

\begin{claim}\label{claim:property-optimal-layers}
    For every $0 \leq t' < t \leq \ell-1$, we have $Z_{t'}^\star \cap Z_{t}^\star = \emptyset$.
\end{claim}

\begin{proof}
    The claim is obvious for $t'=0$.
    Now, assume $1 \leq t' \leq t \leq \ell-1$ and $Z_{t'}^\star \cap Z_{t}^\star \neq \emptyset$.
    This means that there exist $A = (w_0,z_1,w_1,\ldots,z_{\ell-1}, w_{\ell-1},z_{\ell}) \in \calP^\star$ and $A' = (w'_0,z'_1,w'_1,\ldots,z'_{\ell-1}, w'_{\ell-1},z'_{\ell}) \in \calP^\star$ such that $z_t = z'_{t'}$.
    The idea is to merge the first $t'$ blocks of $A'$ together with the last $\ell-t-1$ blocks of $A$ via $z_t = z'_{t'}$ and construct an augmenting chain of length strictly less than $\ell$ out of it.
    To make this formal, define
    $$ S = (w'_0,z'_1,w'_1,\ldots,z'_{t'-1},w'_{t'-1}, z_t, w_t, \ldots, z_{\ell-1},w_{\ell-1},z_{\ell})  . $$
    Since $z'_{t'} = z_t$, it is straightforward to see that $S$ is a pseudo-augmenting chain w.r.t.~$(\calF, \calM, D)$. The claim then follows from \Cref{claim:pseudoChain-into-chain}.
\end{proof}

\subsection{Description of the Algorithm For \Cref{lem:subroutine-main}}\label{sec:subroutine-main-description}

We now explain how our algorithm for \Cref{lem:subroutine-main} works.
We can divide the algorithm into two phases as follows.

\medskip
\noindent
\textbf{Phase I: Find the Layers.}
In this phase, the algorithm tries to find $Z_t^\star$ as defined in \Cref{sec:otpimal-layers} for all $t \in [0, \ell-1]$.
Finding the sets $Z_t^\star$ is challenging since we do not know in advance whether an alternating chain of length $t \leq \ell-1$ can be extended into an augmenting chain of length $\ell$.
Instead, we find a family of mutually disjoint node sets $Z_0, Z_1, \ldots, Z_{\ell-1}$ such that $Z_t^\star \subseteq Z_t$ for every $t \in [0, \ell-1]$.
This means that we are not missing any potential node in $Z_t^\star$.
We refer to $Z_t$ as the $t^\text{th}$ \textbf{layer} throughout our algorithm.
The main idea is to keep track of alternating chains of length $t$ iteratively for $t$ going from $0$ to $\ell-1$ and construct $Z_t$.
We explain this phase in \Cref{sec:alg-phase1}.

\medskip
\noindent
\textbf{Phase II: Find and Apply Augmenting Chains.}
In this phase, we exploit the main property of sets $Z^\star_t \subseteq Z_t$ (see \Cref{claim:property-optimal-layers}) in order to search for augmenting chains.
Since $Z_t^\star \subseteq Z_t$, while searching for augmenting chains, we only need to check nodes of $Z_t$ as potential nodes for the roots of $t^\th$ $\calM$-blocks of augmenting chains of length $\ell$, and we do not miss any node in $ Z_{t}^\star$.
There might be some useless nodes in $Z_t - Z_t^\star$, but we show that the total time spent on these nodes in our algorithm is at most $\tilde{O}(m)$.
We provide a backward DFS-based search that starts from layer $t = \ell-1$ and tries to construct an augmenting chain of length $\ell$ by considering potential nodes in the layers $Z_{\ell-1}, \ldots, Z_1, Z_0$.
We explain this phase in \Cref{sec:alg-phase2}.

\subsubsection{Phase I: Find the Layers}\label{sec:alg-phase1}

Now, we explain how to construct the layers $Z_t$ iteratively.
Initially, the $0^\text{th}$ layer $Z_0$ consists of the dummy root of all special $\calM$-molecules, which is identical to $Z^\star_0$.
We maintain a label for each node as either \textbf{scanned} or \textbf{unscanned}.
Initially, all nodes are unscanned.
A node $u$ becomes scanned while constructing layer $Z_t$ for only one reason; there exists an $x \in Z_{t-1}$ such that $u$ is in an $\calM$-block rooted at $x$.
For instance, during the construction of $Z_1$, all of the nodes contained in special $\calM$-molecules become scanned.

\medskip
\noindent
\textbf{Constructing Layer $Z_{t+1}$.}
Now, we explain how to construct $Z_{t+1}$ assuming that we have already constructed $Z_0,\ldots, Z_{t}$.
The $(t+1)^\text{th}$ layer $Z_{t+1}$ will be constructed as follows, which is aligned with the definition of alternating chain. The idea is to try and expand alternating chains of length $t$ into alternating chains of length $t+1$.

We iterate over all $x \in Z_t$ ($x$ may be treated as the the node $z_{t}$ in some alternating chain $(w_0,z_1,w_1,\ldots,z_t,w_t)$), mark $x$ as scanned, then iterate over all unscanned nodes in $T^\calF_x$.
Assume $u \in V(T^\calF_x)$ is such an unscanned node.
First, we mark $u$ as scanned.
If $u$ is not contained in an $\calM$-atom, we skip $u$ and move to the next unscanned node in $T^\calF_x$.
The reason is that $u$ represents $w_t$ in an alternating chain, and according to \Cref{def:alt-chain}, $w_t$ must be $\calM$-reducible.
Otherwise ($u$ is inside an atom), we iterate over all neighbors of $u$ like $v$ that is a potential node for $z_{t+1}$ in an alternating chain. 
We then put $v$ into $Z_{t+1}$ if it satisfies the following conditions;
$$(u, v) \in E - E(\calF), \text{ and } v \notin V(T^\calF_{u \leftarrow x}). $$
Once we scan all $x \in Z_t$, all $u \in V(T^\calF_x)$, and all $v \in \psi_G(u)$ (neighbors of $u$), we have $Z_{t+1}$ and pass it to the pruning procedure as follows.

\medskip
\noindent
\textbf{Pruning Layer $Z_{t+1}$.}
At the end of the construction of $Z_{t+1}$,  we consider every $x \in Z_{t+1}$ and remove $x$ from $Z_{t+1}$ if at least one of the following holds.
\begin{itemize}
    \item $x$ is already scanned.

    \item $x$ is neither $\calM$-covered nor the root of a normal $\calM$-molecule.
\end{itemize}
The reason for this pruning is as follows.
If $x$ is already scanned, the sub-tree $T^\calF_x$ is scanned as well, and we show later that in this case, $x$ can not appear as $z_{t+1}$ in an augmenting chain of length $\ell$.
The idea is similar to the proof of the main property in \Cref{claim:property-optimal-layers}.
If $x$ satisfies the second condition above, there is no $\calM$-block rooted at $x$, and $x$ can not obviously appear as $z_{t+1}$ in an augmenting chain of length $\ell$ according to \Cref{def:aug-chain}.

\subsubsection{Phase II: Find and Apply Augmenting Chains}\label{sec:alg-phase2}

We now use the layers $Z_t$ in order to find augmenting chains.
First, we look for potential last edges $(w_{\ell-1},  z_{\ell})$ that can appear in an augmenting chain, and once we find such an edge, we start a backward search from layer $\ell-1$ to layer $0$.
Every time we find an augmenting chain, we apply it to $(\calF, \calM, D)$ and keep track of what happens to the forest and its molecular decomposition.
Now, we explain the procedure in the following.

\medskip
\noindent
\textbf{Searching For $(w_{\ell-1}, z_{\ell})$.}
The procedure of searching for $(w_{\ell-1}, z_{\ell})$ is almost identical to what we did for constructing the layers.
We iterate over all $x \in Z_{\ell-1}$, mark $x$ as scanned, and iterate over all unscanned nodes $u$ included in an $\calM$-atom inside $T^\calF_x$.
Let $u \in V(T^\calF_x)$ be such a node.
We mark $u$ as scanned, and iterate over all neighbors of $u$ like $v$.
If $v$ satisfies the following conditions, we start a backward search from $(w_{\ell-1}, z_{\ell}) := (u, v)$;
\begin{itemize}
    \item $(u, v) \in E - E(\calF)$,
    \item $v \notin V(T^\calF_{u \leftarrow x})$, and
    \item $v$ is either \{$\calM$-reducible\} or \{$\calM$-free with $\deg_{\calF}(v) \leq \Delta^\star$ and $v \notin D$\}.
\end{itemize}
The goal of this backward search is to search for an augmenting chain ending at $(u, v)$ and apply it to $(\calF,\calM,D)$.
It is either successful or unsuccessful
In both cases, after the termination of backward search, we continue our search for the next potential edge $(w_{\ell-1}, z_{\ell})$ from the last time that we started a backward search.
The difference is that if the backward search successfully finds an augmenting chain, the configuration $(\calF, \calM, D)$ is updates, and the rest of the search for potential edges $(w_\ell, z_{\ell+1})$ is w.r.t.~the current new configuration $(\calF, \calM, D)$.

\medskip
\noindent
\textbf{Backward Search.}
The backward search starts from a non-forest edge $(w_{\ell-1}, z_{\ell})$ and constructs a sequence $(w_t, z_{t+1}, w_{t+1}, \ldots, w_{\ell-1}, z_{\ell})$ going down the layers.
Once it finds a $w_0$ contained in a special $\calM$-molecule, we show that this sequence must be an augmenting chain w.r.t.~the current configuration $(\calF, \calM, D)$.

For efficiency, we consider a label for every node as either \textbf{effective} or \textbf{ineffective}.
Moreover, we consider these labels for each edge of the graph.\footnote{We only use this label for a node $u$ that is contained in at least one layer $Z_t$ for some $t \in [0,\ell]$. Similarly, we only use this label for non-forest edges $(u,v) \in E - E(\calF)$ such that $\{u,v\} \cap Z_t \neq \emptyset$ for some $t \in [1,\ell-1]$.}
Initially, all edges and nodes are effective.
A node $x$ becomes ineffective once we figure out that there is no augmenting chain w.r.t.~the current configuration of length $\ell$ that contains $x$ as the root of its $t^\th$ $\calM$-block (i.e., $z_t$) for some $t \in [0, \ell-1]$.
This label is for speeding up the search.

\medskip
\noindent
\textbf{Search at Layer $t \in [1, \ell-1]$.}
Now, assume that we have $ (w_t,z_{t+1},\ldots,w_{\ell-1}, z_\ell)$ and want to find $(w_{t-1}, z_t)$ and append it to the beginning of this sequence.
The procedure is aligned with the definition of an augmenting chain.

We consider all effective ancestors of $w_t$ in $Z_t$, i.e., all effective nodes $x \in Z_t$ such that $w_t \in T^\calF_x$ and iterate over these nodes $x$ in the decreasing order of their depth (while considering each molecule as a rooted sub-tree of $\calF$).
Then, we iterate over all effective edges incident on $x$ like $(y, x)$.
If $y$ is $\calM$-reducible and $(y, x) \in E -  E(\calF)$, we set $(w_{t-1}, z_{t}) := (y, x)$ and pass $ (w_{t-1},z_{t},\ldots,w_{\ell-1}, z_\ell)$ to the search at layer $t-1$.
Otherwise, (if $y$ does not satisfy these properties), we mark the edge $(y, x)$ as ineffective and move to the next effective edge $(y', x)$.
In the former case, the following can happen to the call at layer $t-1$.
\begin{itemize}
    \item \textbf{The Search at Layer $t-1$ is Successful.}
    In this case, we show that the final sequence $(w_0, z_1, w_1 , \ldots , w_\ell, z_{\ell+1})$ found at layer $0$ is an augmenting chain w.r.t.~the current configuration $(\calF, \calM, D)$ and we terminate the backward search.

    \item \textbf{The Search at Layer $t-1$ is Unsuccessful.}
    In this case, we mark the edge $(y, x)$ as ineffective and continue to the next effective edge $(y', x)$ incident on $x$.
\end{itemize}
If all of the recursive calls at layer $t-1$ for all neighbors of $x$ are unsuccessful, we mark $x$ as ineffective and move to the next effective ancestor of $w_t$ in $Z_t$.
We will show that if $x$ is marked ineffective, it can not appear as the root of $t^\th$ $\calM$-block of (i.e., $z_{t}$) of any augmenting chain w.r.t.~any configuration $(\calF, \calM, D)$ \emph{until the end of Phase II}.

Finally, if all of the effective ancestors $x$ of $w_t$ in $Z_t$ become ineffective in this call, we say that this search at layer $t-1$ is unsuccessful.

\medskip
\noindent
\textbf{Search at Layer $t=0$.}
When a call at layer $t=0$ is made, we have the entire sequence $(w_0, z_1, w_1 , \ldots , w_\ell, z_{\ell+1})$.
We simply check whether $w_0$ is inside a special $\calM$-molecule.
If this is the case, we say that the backward search is succussful, and we show that this sequence is indeed an augmenting chain w.r.t.~the current configuration $(\calF, \calM, D)$.

Next, we update $(\calF, \calM, D)$ by applying this augmenting chain to it, and terminate the entire backward search.
The procedure of applying an augmenting chain to $(\calF, \calM, D)$ is explained in \Cref{sec:molecules-atoms}, and it is completely algorithmic.

The current configuration $(\calF, \calM, D)$ will change at the end of a successful backward search, and the algorithm continues searching for more augmenting chains w.r.t.~the new updated configuration.

\subsection{Analysis of Our Algorithm for \Cref{lem:subroutine-main}}\label{sec:subroutine-main-analysis}

\subsubsection{Correctness of Phase I}\label{sec:analysis-phase-I}
We show the following series of claims.

\begin{claim}\label{claim:Zt-appears-in-partial-chains}
    For every $t \in [1, \ell-1]$ and every node $v \in Z_t$, there exists an alternating chain $ P = (w_0, z_1, w_1, \ldots, z_{t-1}, w_{t-1})$ of length $t-1$ w.r.t.~$(\calF, \calM, D)$ (the initial input) such that $v \in \tail(P)$.
\end{claim}

\begin{claim}\label{claim:Zi-Zj-are-disjoint}
    $Z_{t'} \cap Z_{t} = \emptyset$ for all $0 \leq t' < t \leq \ell-1$.
\end{claim}

\begin{claim}\label{claim:Zt-star-subseteq-Zt}
    For every $t \in [0,\ell-1]$, we have $Z_t^\star \subseteq Z_t$.
\end{claim}

The proofs of these claims are deferred to \Cref{sec:deferred-proofs}.
\Cref{claim:Zt-appears-in-partial-chains} shows the intuition of the objective of Phase I of the algorithm.
Basically, the algorithm is trying to expand alternating chains during Phase I, and it makes sense to do so since every $z_t \in Z_t^\star$ appears as the root of the $t^\th$ $\calM$-block of an alternating chain.
Claims \ref{claim:Zi-Zj-are-disjoint} and \ref{claim:Zt-star-subseteq-Zt} show that the layers found by the algorithm in Phase I form a family of mutually disjoint sets that contains all potential nodes that can be the root of the $t^\th$ $\calM$-block of an alternating chain, and we are not missing any node.

\subsubsection{Correctness of Phase II}\label{sec:analysis-phase-II}
Since after finding each augmenting chain in Phase II, $(\calF, \calM, D)$ is updated, we assume that a total of $q$ many augmenting chains $(A^{(1)}, A^{(2)}, \ldots, A^{(q)})$ is found during Phase II, and $(\calF^{(i)}, \calM^{(i)}, D^{(i)})$ is the result of applying $A^{(i)}$ to $(\calF^{(i-1)}, \calM^{(i-1)}, D^{(i-1)})$ (initially $(\calF^{(0)}, \calM^{(0)}, D^{(0)}) := (\calF, \calM, D)$ is the input).
To prove the correctness of the algorithm, we need to show that;
\begin{itemize}
    \item Every sequence $A^{(i)} = \left(\{(z_r,w_r)\}_{r=0}^\ell, z_{\ell+1}\right)$ found by the algorithm is a valid augmenting chain w.r.t.~$(\calF^{(i-1)}, \calM^{(i-1)}, D^{(i-1)})$, and
    \item At the end of Phase II, there is no augmenting chain of length at most $\ell$ w.r.t.~$(\calF^{(q)}, \calM^{(q)}, D^{(q)})$.
\end{itemize}
We show these two main propositions in a series of claims as follows.

\begin{claim}\label{claim:ineffective-correctness}
    For any $i \in [0, q]$ and $t \in [1, \ell-1]$, if a node $x \in Z_t$ is marked ineffective by the algorithm while searching for augmenting chains w.r.t.~$(\calF^{(i)}, \calM^{(i)}, D^{(i)})$, then there is no augmenting chain $A = (w_0,z_1,w_1,\ldots,z_{\ell-1},w_{\ell-1},z_{\ell})$ of length $\ell$ w.r.t.~$(\calF^{(i)}, \calM^{(i)}, D^{(i)})$ satisfying $x = z_t$.
\end{claim}

\begin{claim}\label{claim:backward-search-correctness}
    When the backward search finds the $i^\text{th}$ sequence $A^{(i)} = (w_0,z_1,w_1,\ldots, z_{\ell-1}, w_{\ell-1}, z_\ell)$, it is a valid augmenting chain w.r.t.~$(\calF^{(i-1)}, \calM^{(i-1)}, D^{(i-1)})$.
\end{claim}

\begin{claim}\label{claim:no-more-chain-at-the-end-of-Phase-II}
    At the end of Phase II, $(\calF^{(q)}, \calM^{(q)}, D^{(q)})$ is an $(\ell+1)$-configuration, i.e., there does not exist any augmenting chain of length at most $\ell$ w.r.t.~$(\calF^{(q)}, \calM^{(q)}, D^{(q)})$.
\end{claim}

The proofs of these claims are deferred to \Cref{sec:deferred-proofs}.
Recall \Cref{claim:set-of-augmenting-chains-is-decreasing} from \Cref{sec:alt-aug-chains}.
It indicates that we do not need to reconstruct the layers $Z_0, Z_1, \ldots, Z_\ell$ in order to search for augmenting chains of length $\ell$ after each update on $(\calF, \calM, D)$.
\Cref{claim:ineffective-correctness} also shows the correctness of the labels.
As a result, considering the nodes that are still marked as effective and are contained in the layers $Z_t$ suffices to search for augmenting chains of length $\ell$ w.r.t.~the current configuration $(\calF, \calM, D)$.
Finally, Claims \ref{claim:backward-search-correctness} and \ref{claim:no-more-chain-at-the-end-of-Phase-II} conclude that the final $(\calF^{(q)}, \calM^{(q)}, D^{(q)})$ is indeed an $(\ell+1)$-configuration, which completes the correctness of the algorithm.

\subsubsection{Running Time (Implementation Details)}

\medskip
\noindent
\textbf{Time Spent on Constructing Layers.}
Recall the labels `scanned' and `unscanned' for the nodes.
Assume $x \in Z_t$ and we want to scan the sub-tree $T^\calF_x$.
We can easily determine the set of unscanned nodes inside $T^\calF_x$ by maintaining a label for each node $y \in V$ that indicates whether or not $T^\calF_y$ is scanned.
As a result, in order to iterate over the set of unscanned nodes inside $T^\calF_x$, we can do a simple DFS, and whenever we reach a node $y$ such that $T^\calF_y$ is already scanned, we ignore the entire sub-tree $T^\calF_y$ within $T^\calF_x$.

This concludes that the running time of scanning the nodes in $T^\calF_x$ is actually proportional to the number of unscanned nodes in there, i.e., $O(1 + \#\{\text{unscanned nodes in } T^\calF_x\})$.
According to \Cref{claim:Zi-Zj-are-disjoint}, we consider each subtree $T^\calF_x$ at most once for every fixed node $x$.
Hence, the total time for scanning nodes during the construction of the layers is $\tilde{O}(n)$.

For every unscanned node $u$ that is contained in an atom of $T^{\calF}_x$, we iterate over all neighbors $v \in \psi_G(u)$ only once.
For every neighbor $v \in \psi(u)$, we must check whether $v \in V(T^{\calF}_{u \leftarrow x})$.
This property can be easily performed in $\tilde{O}(1)$ time by simple data structures on rooted trees.
More precisely, we maintain the depth of nodes (according to their depth from roots of molecules to leaves), and save the ancestors of depth $(\text{depth}(x) - 2^i)$ for each $i \geq 0$.
Then we can find the unique ancestors of $u$ and $v$ at depth $\text{depth}(x) + 1$ in $\tilde{O}(1)$ time.
These ancestors are equal if and only if $v \in V(T^\calF_{u \leftarrow x})$.

As a result, the total time spent for constructing the layers $Z_0, Z_1, \ldots, Z_{\ell-1}$ is at most $\tilde{O}(m)$.

\medskip
\noindent
\textbf{Time Spent on the Pruning Procedure.}
In the pruning procedure of $Z_{t}$, some of the nodes in $Z_{t}$ might be removed from this set.
We can simply charge the running time of removing the node $v$ from $Z_t$ to the edge $(u, v)$ that makes the algorithm put $v$ into $Z_t$.
All of these edges are distinct according to the argument in the previous paragraph.
As a result, the total time spent in Phase I of the algorithm is at most $\tilde{O}(m)$.

\medskip
\noindent
\textbf{Time Spent on Searching For $(w_{\ell-1}, z_{\ell})$.}
Since the algorithm for searching the edges $(w_{\ell-1}, z_{\ell})$ is similar to what it does for constructing the layers, the running time of this part is $\tilde{O}(m)$ as well.

\medskip
\noindent
\textbf{Time Spent on Finding Potential Effective Ancestors in the Backward Search.}
Recall the labels `effective' and `ineffective' for the nodes.
Consider a $t \in [0, \ell-1]$ and a $w_t$ inside a call at layer $t$ of the backward search.
Below, we describe how to efficiently iterate over effective ancestors of $w_t$ in $Z_t$, i.e., effective nodes $x \in Z_t$ such that $w_t \in V(T^\calF_x)$.

We maintain a rooted forest $\calF(Z_t)$ w.r.t.~$Z_t$ as follows;
\begin{itemize}
    \item There are exactly $|Z_t|$ non-leaf nodes in $\calF(Z_t)$, each corresponding to a unique node in $Z_t$.
    
    \item The leaves of $\calF(Z_t)$ correspond to $\calM$-atoms that are contained in at least one sub-tree $T^\calF_x$ for some $x \in Z_t$.

    \item For each leaf $v$ in $\calF(Z_t)$ corresponding to the $\calM$-atom $C_v$, the parent of $v$ in $\calF(Z_t)$ is the least ancestral node $x \in Z_t$ such that $V(C_v) \subseteq V(T^\calF_x)$.

    \item For each non-leaf node in $\calF(Z_t)$ corresponding to $x \in Z_t$, the parent of $x$ in $\calF(Z_t)$ is the least ancestral ancestor of $x$ that is contained in $Z_t$.

    \item Each node in $\calF(Z_t)$ is either `effective' or `ineffective'.
\end{itemize}
Note that we can construct these data structures for all $t \in [1, \ell-1]$ simultaneously by performing only one DFS call at the end of Phase I.

Now, once we have a node $w_t$ in an $\calM$-atom corresponding to a leaf node $v$ in $\calF(Z_t)$, we can efficiently iterate over the effective ancestors of $w_t$ in $Z_t$ as follows.
The algorithm considers the ancestors of $w_t$ in the decreasing order of their depth (w.r.t.~the rooted molecules), and marks them ineffective as long as no augmenting chain is found.
As a result, while considering the ancestors of an $\calM$-atom $C$ in $\calF(Z_t)$, if we reach an ancestor $x \in Z_t$ which is already marked `ineffective', we can stop searching for the rest of the ancestors of $C$.
The reason is that there must be another $\calM$-atom $C'$ which satisfies $V(C') \subseteq V(T^\calF_x)$ and $x$ is marked `ineffective' because of $C'$.
But, in this case, all of the ancestors of $x$ in $Z_t$ must also become `ineffective' because of $C'$.
If this is not the case, an augmenting chain must have been found with $C'$ being one of its critical $\calM$-atom, and both of the $\calM$-atoms $C$ and $C'$ (which are in the same $\calM$-molecule) will become free nodes, and the algorithm no longer considers $C$ as an $\calM$-atom.
Hence, the algorithm can efficiently iterate over all effective ancestors of an $\calM$-atom in $Z_t$.

\medskip
\noindent
\textbf{Number of Times a Node is Considered a Potential Ancestor.}
Assume $x \in Z_t$ is considered as a potential ancestor of $w_t$.
For every effective edge $(x, y)$, there might be a recursive backward search at layer $t-1$.
If this recursive call is unsuccessful, the edge triggering the call will become ineffective and will never be scanned again.
If the recursive call is successful, it means that $y$ must be inside an atom, and after applying the augmenting chain $y$ becomes a free node.
Hence, the next time that the edge $(x,y)$ is scanned, it will become ineffective.
Hence, the total number of times that a node $x$ can be considered as a potential ancestor of some node throughout the entire Phase II is at most $n_x = O(\deg_{\calF^{(0)}}(x))$.

\medskip
\noindent
\textbf{Time Spent on Exploring Potential Effective Edges.}
Now, consider a node $x \in Z_t$ as a potential ancestor of $w_t$.
According to the previous explanations, the effective edges incident on $x$ will be scanned at most twice.
As a result, the total time spent on searching for effective edges $(x, w_{t-1})$ for each fixed $x$ is at most $O(n_x + \deg_{\calF^{(0)}}(x))$.
Finally, by summing up these values for all $x$, we conclude that the total time spent on searching for augmenting chains in Phase II of the algorithm is at most $\tilde{O}(m)$.

\medskip
\noindent
\textbf{Time Spent on Applying Augmenting Chains to Configurations.}
Consider an augmenting chain $A = (w_0,z_1,w_1,\ldots,w_{\ell-1},z_\ell,) $.
The time spent on the degree-reduction subroutine of $w_i$ (and possibly $z_{\ell}$ if it is contained in an $\calM$-atom) is at most $\tilde{O}(m_{w_i})$ according to \Cref{lem:degree-reduction}, where $m_{w_i}$ is the number of edges in the $\calM$-atom containing $w_i$.
Since the atoms containing $w_i$s are all distinct among all augmenting chains that the algorithm finds, we conclude that the total time spent on the degree-reduction subroutine is at most $\tilde{O}(m)$ throughout the entire Phase II.

Finding the child $y_i$ of $z_i$ that satisfies $T^\calF_{w_i \leftarrow z_i} = T^\calF_{y_i \leftarrow z_i}$ can be done easily by maintaining simple data structures on rooted trees.
The only important note here is that we do not update the data structures after changing the forest.
All of these searches can be performed according to the very initial forest $\calF$ given as input.
The reason is that the molecules containing $w_0,w_1,\ldots,w_\ell$ and possibly $z_{\ell+1}$ are not affected at the time of applying $A$, and the edges $(w_i,z_{i+1}), (y_i,z_i)$ coincide in the forest at that time and the initial forest $\calF$.
Hence, we initialize these data structures only once.
This argument also shows that the total number of edges that we change in the forest by applying all of the augmenting chains is at most $\tilde{O}(m)$ (because each edge appears in at most one augmenting chain in Phase II).
Updating $\calM$ and $D$ trivially takes at most $\tilde{O}(n)$ time.

\subsection{Deferred Proofs}\label{sec:deferred-proofs}

\subsubsection{Proof of \Cref{claim:Zt-appears-in-partial-chains}}

We prove the claim inductively.
Assume $v \in Z_t$.
According to the procedure of constructing $Z_t$, we conclude that there must exist $x \in Z_{t-1}$ and $u$ contained in an $\calM$-atom inside $T^\calF_x$ such that $(u,v) \in E - E(\calF)$ and $v \notin  V(T^\calF_{u \leftarrow x})$.
Since $x \in Z_{t-1}$, we have that there exist an alternating chain $ P = (w_0, z_1, w_1, \ldots, z_{t-2}, w_{t-2}) $ of length $t-2$ w.r.t.~$(\calF, \calM, D)$ such that $x \in \tail(P)$ (in the case $t=1$, we have that $x$ is the dummy root of an special $\calM$-molecule).
Let $P' : = (w_0,z_1,w_1,\ldots,z_{t-2}, w_{t-2}, x, u)$.
According to the properties of $v$ it is obvious that $v \in \tail(P')$.
It suffices to show that $P'$ is an alternating chain w.r.t.~$(\calF, \calM, D)$.
Properties \ref{property-z0} and \ref{property-non-forest-edge} in \Cref{def:alt-chain} are trivial.
Since $P$ is an alternating chain, it suffices to show the following properties for $P'$.

\medskip
\noindent
\textbf{Property \ref{property-wi-outside-jth-block} for $0 \leq j < i= t-1$.}
This property follows from the description of the algorithm.
Since the node $u$ (and the atom in which it is contained) is scanned during the construction of $Z_t$, it is not contained in the $\calM$-block $T^\calF_{w_j \leftarrow z_j}$, as otherwise it would have been scanned in the previous iteration during the construction of $Z_{j+1}$.

\medskip
\noindent
\textbf{Property \ref{property-wi-zi-connection} for $i = t-1$.}
The only non-trivial part is part $(b)$.
We have that $x$ must be the root of an $\calM$-molecule or be contained in an $\calM$-molecule, as otherwise $x$ must have been removed from $Z_{t-1}$ during the pruning procedure of $Z_{t-1}$ at the end of the construction of $Z_{t-1}$ in Phase I.
So, $T_{u \leftarrow x}^\calF$ is contained in an $\calM$-molecule.
Moreover, if $x$ is inside an $\calM$-molecule and is $\calM$-reducible, then $(P,x)$ will be an augmenting chain of length $t-1 < \ell$, which is a contradiction.
We conclude that $x$ is either the root of an $\calM$-molecule or is $\calM$-non-reducible.
Hence, $T_{u \leftarrow x}^\calF$ is $\calM$-block.

\subsubsection{Proof of \Cref{claim:Zi-Zj-are-disjoint}}

During the construction of $Z_t$, since $t' < t$, all of the nodes in $Z_{t'}$ are already scanned.
Hence, if a node $x \in Z_{t'}$ was put into $Z_t$, it would be removed during the pruning procedure at the end of the construction of $Z_{t}$ in Phase I.

\subsubsection{Proof of \Cref{claim:Zt-star-subseteq-Zt}}

We show this claim inductively.
It is obvious that $ Z_0^\star = Z_0$ since both of them equal the set of dummy roots of the special $\calM$-molecules.
Now, assume that $Z_{t'}^\star \subseteq Z_{t'}$ for all $t' \in [0, t-1]$. We will prove it for $t$.
Assume $y \in Z_t^\star$ is arbitrary.
According to the definition of $Z_{t}^\star$, we conclude that there exists an augmenting chain $A = (w_0,z_1,w_1,\ldots,z_{\ell-1}, w_{\ell-1}, z_\ell)$ w.r.t.~$(\calF, \calM, D)$ such that $y = z_t$.
We have that $z_{t-1} \in Z_{t-1}^\star \subseteq Z_{t-1}$.
We show that $w_{t-1}$ was not scanned before starting the construction of $Z_t$ and becomes scanned during the construction of $Z_t$.

\begin{claim}
    Assume $w_{t-1}$ (and the atom containing $w_{t-1}$) becomes scanned for the first time while scanning $T_{x}^\calF$ for some $x \in Z_{t'}$ during the construction of $Z_{t'+1}$.
    Then, we must have $t' = t-1$.
\end{claim}

\begin{proof}
    First, we show $t' \leq t-1$, i.e., $w_{t-1}$ becomes scanned no later than the construction $Z_t$. 
    Consider the execution of the algorithm while constructing $Z_{t}$.
    In some iteration, $z_{t-1} \in Z_{t-1}$ will be considered and the sub-tree $T^\calF_{z_{t-1}}$ of $z_{t-1}$ will be scanned.
    $w_{t-1}$ is contained in an $\calM$-atom inside $T^\calF_{z_{t-1}}$ since $A$ is an augmenting chain.
    Hence, either $w_{t-1}$ is already scanned, or becomes scanned at this time.
    
    Now, we show that $t' \geq t-1$.
    Since $x \in Z_{t'}$, according to \Cref{claim:Zt-appears-in-partial-chains}, we conclude that there exists an alternating chain $ P' := (w_0', z_1', w_1', \ldots, z'_{t'-1}, w'_{t'-1})$ w.r.t.~$(\calF, \calM, D)$ such that $x \in \tail(P')$.
    Now, define
    $$ S := (w_0',z_1',w_1',\ldots,z'_{t'-1}, w'_{t'-1},x,w_{t-1},z_t,w_t,\ldots, z_{\ell-1}, w_{\ell-1}, z_\ell). $$
    Since $A$ is an augmenting chain, $P'$ is an alternating chain, $t \leq \ell-1$, $x \in \tail(P')$, and $w_{t-1} \in T_{x}^\calF$, it is straightforward to see that $S$ is a pseudo-augmenting chain w.r.t.~$(\calF, \calM, D)$ of length $\ell - t + t'+1$.
    \Cref{claim:pseudoChain-into-chain} concludes that there exists an augmenting chain w.r.t.~$(\calF, \calM, D)$ of length at most $\ell - t + t'+1$.
    Since $(\calF, \calM, D)$ is an $\ell$-configuration, we conclude that $t' \geq t-1$.
\end{proof}

Now, assume that $w_{t-1}$ becomes scanned while scanning $T_x^\calF$ for some $x \in Z_{t-1}$ (according to the above claim).
We will show that when the algorithm considers the neighbors of $w_{t-1}$, it will put $z_t$ into $Z_t$. 
The condition $(w_{t-1}, z_t) \in E - E(\calF)$ holds since since $A$ is an augmenting chain.
It remains to show that $z_t \notin V(T^\calF_{w_{t-1} \leftarrow x})$ to conclude that $z_t$ will be put into $Z_t$ by the algorithm.

\begin{claim}
    $z_t \notin V(T^\calF_{w_{t-1} \leftarrow x}$).
\end{claim}

\begin{proof}
    We show something stronger, that is $z_t \notin V(T_x^\calF)$.
    For the sake of contradiction, assume that $z_t \in V(T_{x}^\calF)$.
    Since $t \leq \ell-1$ and $A$ is an augmenting chain, we have $w_t \in V(T^\calF_{z_t}) \subseteq V(T_x^\calF)$.
    Since $x \in Z_{t-1}$, according to \Cref{claim:Zt-appears-in-partial-chains}, there exists an alternating chain $ P' := (w'_0,z'_1,w_1',\ldots,z_{t-2}',w_{t-2}')$ w.r.t.~$(\calF, \calM, D)$ such that $x \in \tail(P')$.
    Now, define
    $$ S := (w_0',z_1',w_1',\ldots,z_{t-2}',w_{t-2}',x,w_t,z_{t+1},w_{t+1},\ldots,z_{\ell-1},w_{\ell-1}, z_{\ell}). $$
    According to $A$ being an augmenting chain, $P'$ being an alternating chain, $x \in \tail(P')$, and $w_t \in V(T^\calF_x) $, it is straightforward to see that $S$ is a pseudo-augmenting chain w.r.t.~$(\calF, \calM, D)$ of length $\ell -1$.
    \Cref{claim:pseudoChain-into-chain} concludes that there exists an augmenting chain of length at most $\ell - 1$ w.r.t.~$(\calF, \calM, D)$.
    This is a contradiction.
\end{proof}

Now, we have that $z_t$ is put into $Z_{t}$ by the algorithm.
It remains to show that $z_t$ will not be removed from $Z_t$ during the pruning procedure, which completes the proof of \Cref{claim:Zt-star-subseteq-Zt}.

\begin{claim}
    $z_t$ is not removed from $Z_t$ during the pruning procedure of $Z_t$.
\end{claim}

\begin{proof}
    For the sake of contradiction, assume that $z_t$ is removed from $Z_t$ during the pruning procedure.
    Since $A$ is an augmenting chain and $t \in [1, \ell-1]$, we have that $z_t$ is the root of an $\calM$-molecule or is $\calM$-covered.
    Hence, the only reason that the algorithm removes $z_t$ from $Z_t$ is that $z_t$ has already been scanned.
    This concludes that there exists $y \in Z_{t'}$ for some $t' \in [0, t-1]$ such that $z_t \in T_y^\calF$.
    According to \Cref{claim:Zt-appears-in-partial-chains}, there exists an alternating chain $ P' := (w'_0,z'_1,w_1',\ldots,z_{t'-1}',w_{t'-1}')$ w.r.t.~$(\calF, \calM, D)$ such that $y \in \tail(P') $.
    Now, define
    $$ S := (w_0',z_1',w_1',\ldots,z_{t'-1}',w_{t'-1}',y,w_t,z_{t+1},w_{t+1},\ldots,z_{\ell-1},w_{\ell-1},z_\ell). $$
    According to $A$ being an augmenting chain, $P'$ being an alternating chain, $y \in \tail(P')$, and $w_t \in V(T^\calF_{z_t}) \subseteq V(T_y^\calF) $, it is straightforward to see that $S$ is a pseudo-augmenting chain w.r.t.~$(\calF, \calM, D)$ of length $\ell-t+t'$.
    \Cref{claim:pseudoChain-into-chain} concludes that there exists an augmenting chain w.r.t.~$(\calF, \calM, D)$ of length at most $\ell-t+t' \leq \ell - 1$.
    This is a contradiction.
\end{proof}

\subsubsection{Proof of \Cref{claim:ineffective-correctness}}

We start with the following claim.

\begin{claim}\label{claim:Zt-is-not-in-atom}
    For every $t \in [1,\ell-1]$, none of the nodes $Z_t$ are contained in an $\calM^{(0)}$-atom.
\end{claim}

\begin{proof}
For the sake of contradiction, assume that $x \in Z_t$ is contained in an $\calM^{(0)}$-atom (i.e., is $\calM^{(0)}$-reducible).
According to \Cref{claim:Zt-appears-in-partial-chains}, since $x \in Z_t$, there exists an alternating chain $P = (w_0,z_1,w_1,\ldots,z_{t-1},w_{t-1})$ such that $x \in \tail(P)$.
Since $x$ is $\calM^{(0)}$-reducible, $(P, x)$ must be an augmenting chain of length $t < \ell$ w.r.t.~$(\calF^{(0)}, \calM^{(0)}, D^{(0)})$.
This is in contradiction with the assumption that the initial configuration given to the algorithm is an $\ell$-configuration.
\end{proof}

Now, we proceed with the proof of \Cref{claim:ineffective-correctness}.
For the sake of contradiction, assume that the claim does not hold for at least one $t \in [1, \ell-1]$.
Consider the smallest value $t^\star$ that violates the claim.
Assume $x$ is marked ineffective and $x = z_{t^\star}^\star$ for some augmenting chain $A = (w_0^\star,z_1^\star,w_1^\star,\ldots,z_{\ell-1}^\star, w_{\ell-1}^\star, z_\ell^\star)$ w.r.t.~$(\calF^{(i)}, \calM^{(i)}, D^{(i)})$.
According to \Cref{claim:set-of-augmenting-chains-is-decreasing} $A$ is also an augmenting chain w.r.t.~$(\calF^{(0)}, \calM^{(0)}, D^{(0)})$, and we conclude that $z_r^\star \in Z_r$ for all $r \in [1, \ell-1]$ according to \Cref{claim:Zt-star-subseteq-Zt}.
Consider the execution of the backward search when $x$ is marked ineffective.
Let $(w_{t^\star}, z_{t^\star+1},\ldots,w_{\ell-1},z_\ell)$ be the sequence that is passed to the call at layer $t^\star$ in the backward search, where $x = z_{t^\star}^\star$ is marked ineffective in this call.
We must have $w_{t^\star} \subseteq V(T^{\calF^{(i-1)}}_x) = V(T^{\calF^{(i-1)}}_{z^\star_{t^\star}})$, as otherwise, $x$ would not have been considered in this call.
Now, define
$$ S = (w_0^\star,z_1^\star,w_1^\star,\ldots,z^\star_{t^\star-1},w^\star_{t^\star-1}, x,w_{t^\star},z_{t^\star+1},w_{t^\star+1},\ldots, z_{\ell-1}, w_{\ell-1},z_\ell). $$

\begin{claim}
    $S$ is a pseudo-augmenting chain w.r.t.~$(\calF^{(i)},\calM^{(i)},D^{(i)})$
\end{claim}

\begin{proof}
The only non-trivial property is property \ref{property-wi-zi-connection} (b) in \Cref{def:alt-chain}.
$T^{\calF^{(i)}}_{w_r^\star \leftarrow z_r^\star}$ satisfies property \ref{property-wi-zi-connection} (b) for all $r \in [0,t^\star]$ since $A$ is an augmenting chain w.r.t.~$(\calF^{(i)}, \calM^{(i)}, D^{(i)})$.
$T^{\calF^{(i)}}_{w_r \leftarrow z_r}$ for $r \in [t^\star+1, \ell-1]$ also satisfy property \ref{property-wi-zi-connection} (b) for the following reason;
According to the procedure of backward search $w_r$ must be contained in an $\calM^{(i)}$-atom inside $T^{\calF^{(i)}}_{z_r}$.
We conclude that $z_r$ is either the root of an $\calM^{(i)}$-molecule or is $\calM^{(i)}$-covered.
It remains to show that $z_r$ is not contained in any $\calM^{(i)}$-atom.
Assume that $z_r$ is in an $\calM^{(i)}$-atom.
Since every $\calM^{(i)}$-atom is an $\calM^{(0)}$-atom as well, we conclude that $Z_r$ contains a node in an $\calM^{(0)}$-atom, which is in contradiction with \Cref{claim:Zt-is-not-in-atom}.
\end{proof}

Now, we show that the backward search will eventually reach layer $0$ and will never mark $x$ as ineffective.
The reason is that; 1) at the current time step, the sequence $(w_{t^\star}, z_{t^\star+1}, w_{t^\star+1}, \ldots, w_{\ell-1}, z_\ell)$ is passed to layer $t^\star$ and $x$ is being processed, 2) according to minimality of $t^\star$, all of the nodes $z_{r}^\star$ for $r \in [0, t^\star-1]$ are marked effective and remain effective (until we transfer to $(\calF^{(i+1)}, \calM^{(i+1)}, D^{(i+1)})$ or the termination of the algorithm for $i=q$), and 3) for every $r \in [0, t^\star-1]$ we have $z_r^\star \in Z_r$.
Note that the backward search might find a different sequence than $S$, but $S$ will remain a valid option for the backward search until the end of the procedure, and $S$ will not be missed.
This is in contradiction with the assumption that $x$ is marked ineffective in this call.

\subsubsection{Proof of \Cref{claim:backward-search-correctness}}

We start by showing that $A^{(i)}$ is a pseudo-augmenting chain w.r.t.~$(\calF^{(i-1)}, \calM^{(i-1)}, D^{(i-1)})$.
According to the description of the backward search in Phase II, it is straightforward to see that all the properties of a pseudo-augmenting chain hold, and the only non-trivial property is property \ref{property-wi-zi-connection} (b).
For every $r \in [0, \ell-1]$, since $w_r$ is contained in an atom inside $T^{\calF^{(i-1)}}_{z_r}$, we conclude that $z_r$ is either the root of a $\calM^{(i-1)}$-molecule, or it is $\calM^{(i-1)}$-covered.
It remains to show that $z_r$ is not contained in an $\calM^{(i-1)}$-atom.
This is implied by \Cref{claim:Zt-is-not-in-atom} and the fact that every $\calM^{(i-1)}$-atom is an $\calM^{(0)}$-atom as well.
Now, assume that $A^{(i)}$ is not an augmenting chain w.r.t.~$(\calF^{(i-1)}, \calM^{(i-1)}, D^{(i-1)})$.
\Cref{claim:pseudoChain-into-chain} concludes that there exists an augmenting chain $\tilde{A}$ of length \emph{strictly} less than $\ell$ w.r.t.~$(\calF^{(i-1)}, \calM^{(i-1)}, D^{(i-1)})$.
\Cref{claim:set-of-augmenting-chains-is-decreasing} implies that $\tilde{A}$ is also an augmenting chain w.r.t.~the initial configuration $(\calF^{(0)}, \calM^{(0)}, D^{(0)})$.
This is a contradiction, with the initial assumption that $(\calF^{(0)}, \calM^{(0)}, D^{(0)})$ is an $\ell$-configuration.

\subsubsection{Proof of \Cref{claim:no-more-chain-at-the-end-of-Phase-II}}

Assume that there exists an augmenting chain $A = (w_0^\star,z^\star_1,w^\star_1,\ldots,z^\star_{k-1},w^\star_{k-1},z^\star_{k})$ of length $k \leq \ell$  w.r.t.~$(\calF^{(q)}, \calM^{(q)}, D^{(q)})$.
According to \Cref{claim:set-of-augmenting-chains-is-decreasing}, $A$ is an augmenting chain w.r.t.~$(\calF^{(j)}, \calM^{(j)}, D^{(j)})$ for all $j \in [0, q]$ as well.
Hence, $k = \ell$ since $(\calF^{(0)}, \calM^{(0)}, D^{(0)})$ is an $\ell$-configuration, and $z_r^\star \in Z_r$ for all $r \in [1, \ell-1]$ (according to \Cref{claim:Zt-star-subseteq-Zt}).
\Cref{claim:ineffective-correctness} implies that all of the nodes $z_r^\star$ are still marked as effective.
Since $z_{\ell-1}^\star \in Z_{\ell-1}$ and $w_{\ell-1} \in V(T^{\calF^{(i)}}_{z_{\ell-1}^\star})$ for all $i \in [0, q]$, according to the procedure of the algorithm while searching for $(w_{\ell-1}, z_{\ell})$ (before starting the backward search), the node $w_{\ell-1}^\star$ must become scanned at some point while scanning the sub-tree $T^{\calF^{(i^\star)}}_x$ for some $x \in Z_{\ell-1}$ and $i^\star \in [0, q]$.

\begin{claim}
    A backward search must have been started from an edge $( w_\ell^\star, v)$ at the time of scanning $w_\ell^\star$ in the sub-tree $T^{\calF^{(i^\star)}}_x$.
\end{claim}

\begin{proof}
    According to the procedure of Phase II, it is sufficient to show that there exists a node $v$ that satisfies; 1) 
    $(w_\ell^\star,  v) \in E - E(\calF^{(i^\star)})$, 2) $v \notin V(T^{\calF^{(i^\star)}}_{w_{\ell-1}^\star \leftarrow x})$, and 3) $v$ is either \{$\calM^{(i^\star)}$-reducible\} or \{$\calM^{(i^\star)}$-free with $\deg_{\calF^{(i^\star)}}(v) \leq \Delta^\star$ and $v \notin D^{(i^\star)}$\}.
    We show these properties hold for $v = z^\star_{\ell}$.
    
    The first and third properties hold since $A$ is an augmenting chain w.r.t.~$(\calF^{(i^\star)}, \calM^{(i^\star)}, D^{(i^\star)})$.
    If $z_{\ell}^\star \in V(T^{\calF^{(i^\star)}}_{w_{\ell-1}^\star \leftarrow x})$, we conclude that; 1) there is a non-forest edge between $z_{\ell}^\star$ and $w_{\ell-1}^\star$, 2) $z_{\ell}^\star$ and $w_{\ell-1}^\star$ are in the same $\calM^{(i^\star)}$-molecule containing $T^{\calF^{(i^\star)}}_{w_{\ell-1}^\star \leftarrow x}$.
    According to the procedure of defining atoms, we conclude that $z_{\ell}^\star$ and $w_{\ell-1}^\star$ must be in the same $\calM^{(i^\star)}$-atom, which is a contradiction with $A$ being an augmenting chain w.r.t.~$(\calF^{(i^\star)}, \calM^{(i^\star)}, D^{(i^\star)})$.
\end{proof}

Finally, similar to the proof of \Cref{claim:ineffective-correctness}, we can argue that the backward search started from $(w_{\ell-1}^\star, v)$ (note that $v$ is not necessarily equal to $z_{\ell}^\star$) ends successfully.
As a result, the algorithm applies this augmenting chain that contains $w_{\ell-1}^\star$, and the atom containing $w_{\ell-1}^\star$ will no longer be an $\calM^{(i^\star)}$-atom (and also $\calM^{(q)}$-atom accordingly).
This is in contradiction with the assumption that $A$ is an augmenting chain w.r.t.~$(\calF^{(q)}, \calM^{(q)}, D^{(q)})$.

\section*{Acknowledgements}

We thank Chandra Chekuri for pointing us to \Cref{cor:n-dependent}.

Sayan Bhattacharya is funded by the European Union (ERC grant, DYNALP, 101170133). Views and opinions expressed are however those of the author(s) only and do not necessarily reflect those of the European Union or the European Research Council Executive Agency. Neither the European Union nor the granting authority can be held responsible for them.

\bibliographystyle{alpha}
\bibliography{references}

\appendix
\section{Brief Discussion on Other Related Work}\label{sec:technical-overview-comparison}

\subsection{The Algorithm of~\cite{DHZ20}}
The~\cite{DHZ20} algorithm provides an spanning tree of maximum degree $(1+\epsilon)\Delta^\star + O(\epsilon^{-2}\log n)$ in $O(m\epsilon^{-7}\log^7n)$ time for any arbitrary $\epsilon \in (0,1/6)$.
In the following, we provide an example showing that the algorithm can potentially return a spanning tree that has a \textbf{multiplicative} approximation ratio of $\Omega(\log n / \log \log n)$.
First, let us briefly discuss how the~\cite{DHZ20} algorithm works.
The algorithm start with an arbitrary spanning tree $T$.
Then given a threshold $k \in [3, n]$, it tries to reduce the degree of the spanning tree to $k-1$.
Finally, by deliberately choose different thresholds in different rounds, it reduces the maximum degree of the spanning tree until its degree becomes $(1+\epsilon)\Delta^\star + O(\epsilon^{-2}\log n)$.

The main object which is used for this process in~\cite{DHZ20} is called `augmenting sequence', and is defined as follows.
It is a sequence of vertex-disjoint non-tree edges $(w_1,z_1),(w_2,z_2),\ldots, (w_h,z_h) \in E - E(T)$ such that; 1) there exists an $w_0 \in P^T_{w_1,z_1}$ where $\deg_T(w_0) \geq k$, 2) $w_i \in P^{T}_{w_{i+1},z_{i+1}} - \cup_{j=i+2}^h P^{T}_{w_{j},z_{j}}$ for all $i \in [0, h-1]$, and 3) $\deg_{T}(z_i) \leq k-2$ for all $i \in [1, h]$ as well as $\deg_T(w_h) \leq k-2$.
The algorithm improves the spanning tree by considering augmenting sequences of length $h+1$ where $h := 1 + \lceil \log_{1+\epsilon} n \rceil$.

We provide an example of a graph $G$ and a spanning tree $T$ of maximum degree $\Omega(\log n/ \log\log n)$ such that there does not exists any augmenting sequence for any arbitrary threshold $k$.
Moreover the optimal spanning tree $T^\star$ of $G$ has degree $\Delta^\star=3$.
This shows that the algorithm of~\cite{DHZ20} in this specific instance can not improve the spanning tree $T$, that has multiplicative approximation $\Omega(\log n/ \log\log n)$.

We define $G$, $T$, and $T^\star$ inductively.
$G_0$ consists of a single node $r_0$ called \textbf{apex} of $G_0$.
Obviously, $T^\star_0 = \{r_0\}$. 
We also consider $T_0 = \{r_0\}$ a rooted tree.
For any $i \geq 1$;
\begin{enumerate}
    \item Consider a node $r$ as the apex of $G_i$ (as well as the root of $T_i$).

    \item Make $i$ copies of $G_{i-1}$ like $G_{i-1}^{(1)}, G_{i-1}^{(2)}, \ldots, G_{i-1}^{(i)}$.

    \item Connect $r$ via $i$ edges to the apex of these copies of $G_{i-1}$.
    All of these edges are part of $T_i$.
    The only edge that is part of $T^\star_i$ is the edge from $r$ to the apex of $G_{i-1}^{(1)}$.

    \item For each $j \in [1, i-1]$, connect the apex of $G_{i-1}^{(j)}$ and apex $G_{i-1}^{(j+1)}$ via an edge.
    These edges are not part of $T_i$, but they are part of $T^\star_i$.
\end{enumerate}
The following figure illustrates $G_4$, $T_4$, and $T_4^\star$.
\input{figures/fig-G4}

It is straightforward to see that, in the graph $G_q$, we have $\Delta^\star = \max_{x \in V} \deg_{T^\star_q}(x) = 3$, but the maximum degree of $T_q$ is $q$.
Note that $q$ can grow as large as $ \Omega (\log n / \log \log n)$, while the number of nodes in $G_q$ remains at most $ (q+1)! \leq n$ (this can be easily verified according to Stirling's approximation of $(q+1)!$).

\subsection{The Algorithm of~\cite{CQT21}}

There are two main algorithms in \cite{CQT21}, a 
$\approx (1+\varepsilon)\Delta^\star$ approximation for the fractional version of the problem by solving the LP relaxation via the multiplicative weights update (MWU) technique in $\tilde{O}(m)$ time.
The multiplicative $1+\epsilon$ approximation is inherent according to MWU.
The second result is a $(1+\epsilon)\Delta^\star + 2$ approximation in $\tilde{O}(n^2/\epsilon^2)$ time that is achieved as follows; 
After approximating the LP by a factor of $(1 + \epsilon)$, it is possible to sample a sparse subgraph $G' \subseteq G$ of the input such that w.h.p.~there exists a fractional solution of cost at most $(1 + 3\epsilon)(1 + \epsilon)\Delta^\star \leq (1 + 7\epsilon)\Delta^\star$ in $G'$.
According to \cite{SL07}, w.h.p.~there exists an integer solution in $G'$ with max degree $\lceil(1+7\epsilon) \Delta^\star\rceil + 1$.
Finally, by running the \cite{FR92} algorithm on $G'$, \cite{CQT21} get a spanning tree of $G'$ (as well as $G$) with maximum degree $\leq \lceil(1+7\epsilon)\Delta^\star\rceil + 2$ in $\tilde{O}(n^2/\epsilon^2)$ time.

According to this explanation, it is straightforward to see that it is impossible to achieving additive plus one approximation in less than $O(mn)$ time with this algorithm for general input graphs $G$.
First of all, we should consider $\epsilon \leq 1/(7\Delta^\star)$ in order to translate the multiplicative $(1+7\epsilon)\Delta^\star$ approximation to $\Delta^\star + 1$, which increases the running time to $\tilde{O}(n^2\Delta^2)$ and can be as worse as $\Omega(n^3) = \Omega(mn)$ if $\Delta^\star \geq \Omega(\sqrt{n})$.
Moreover, the only guarantee on the sparsified graph $G'$ that we have is that w.h.p.~there exists a spanning tree of maximum degree at most $\lceil(1+7\epsilon)\Delta^\star\rceil + 1 \geq  \Delta^\star+2$.
Hence, even if we can find an optimal spanning tree in the sparsified graph $G'$, its maximum degree might be $\Delta^\star+2$.

\section{Generalization to the Bounded Degree Spanning Tree Problem}\label{appendix:generalization}

Recall the BDST problem as defined in the remarks after \Cref{theorem:main} in \Cref{sec:intro}.
Here, we briefly summarize how our algorithm extends to BDST and provides an additive plus one approximation for BDST.

\medskip
\noindent
\textbf{Modified Definitions and Arguments.}
If we point out the changes in the definition of the main objects, it is straightforward to adjust all the statements according to these new definitions. 
A valid forest $\calF$ must satisfy $\deg_\calF(u) \leq b(u)+1$ for all $u \in V$ instead of $\deg_\calF(u) \leq \Delta^\star+1$.
The definition of a molecule remains unchanged.
However, to define atoms, we need to consider the specific bounds $b(u)$, and the degree-reduction subroutine in \Cref{lem:degree-reduction} updates the atom $C$ containing $u \in V(C)$, achieving $\calF^+$ such that $\deg_{\calF^+}(u) \leq b(u)$ and $\calF^+$ remains a valid forest.
The definitions of alternating and augmenting chains remain unchanged except property \ref{property-zell1-reducible} in \Cref{def:aug-chain}, where we need to have $\deg_{\calF}(z_{\ell+1}) \leq b(z_{\ell+1})$ instead of $\deg_{\calF}(z_{\ell+1}) \leq \Delta^\star$.

All of the statements in the paper can be easily adjusted according to these new definitions, and if we assume that there exists a spanning tree $T^\star$ of $G$ satisfying $\deg_{T^\star}(u) \leq b(u)$ for all $u \in V$, all of the arguments go through.
Here, we only highlight that in the analysis of the main \Cref{key-lemma} (in \Cref{sec:main-lemma}), the potential function that we need to define is $\sum_{x \in W(t)} b(x)$ instead of $\Delta^\star \cdot |W(t)|$.

The main subroutine in the algorithm (as in \Cref{lem:subroutine-main}) also works the same except while searching for the edge $(w_{\ell-1}, z_\ell)$ in Phase II, the condition $\deg_\calF(z_{\ell}) \leq \Delta^\star$ must be replaced by $\deg_\calF(z_\ell) \leq b(z_{\ell})$.

\medskip
\noindent
\textbf{Final Algorithm.}
Our final algorithm works as follows;
If a spanning tree $T^\star$ satisfying $\deg_{T^\star}(u) \leq b(u)$ exists, according to the analysis of the algorithm, it returns a spanning tree $T$ satisfying $\deg_T(u) \leq b(u)+1$ in a total of $\tilde{O}(mn^{3/4})$ time.
But, if no spanning tree $T^\star$ satisfying $\deg_{T^\star}(u) \leq b(u)$ exists, the arguments in the analysis of the algorithm become invalid.
As a result, we must consider a threshold of $\tilde{O}(mn^{3/4})$ for the running time of the algorithm and terminate it if it has not found the desired spanning tree after this running time threshold.
Hence, our algorithm always runs in $\tilde{O}(mn^{3/4})$ time, and if no spanning tree is returned,\footnote{It is easy to see that it is not possible for our algorithm to return a spanning tree which is not valid, i.e., it either returns a valid spanning tree or gets stuck at some iteration and the valid forest can not be further improved.} we can certify that there is no spanning tree $T^\star$ satisfying degree bounds $\deg_{T^\star}(u) \leq b(u)$.

\section{Getting Rid of the Prior Knowledge of $\Delta^\star$}\label{appendix:unknown-Delta-star}

In this section, we show how to run our algorithm without prior knowledge of $\Delta^\star$.

We do a binary search on the value of $\Delta^\star$ using variable $k \in [1, n]$.
We use the generalization of our algorithm for the BDST problem with input $b(u) = k$ for all $u \in V$.
If our algorithm finds a spanning tree of maximum degree $k+1$, we reduce the value of $k$ according to the binary search step.
If our algorithm was unsuccessful in finding a spanning tree of degree at most $k+1$, we increase the value of $k$ according to the binary search step.
Eventually, we consider the smallest value of $k^\star$, where the algorithm has successfully returned a spanning tree of maximum degree at most $k^\star+1$.
We show that $k^\star + 1 \leq \Delta^\star + 1$.

For the sake of contradiction, assume that $k^\star > \Delta^\star$.
In this case, define $k := k^\star - 1 \geq \Delta^\star$ and consider the execution of the algorithm for bounds $b(u) = k$ for all $u \in V$.
Since $k \geq \Delta^\star$, there exists a spanning tree $T^\star$ satisfying $\deg_{T^\star}(u) \leq b(u)$ for all $u \in V$, and our algorithm would successfully return a spanning tree of maximum degree $k+1$.
This is in contradiction with the assumption on the minimality of $k^\star$.

As a result, our algorithm eventually returns a spanning tree of maximum degree $ \leq \Delta^\star+1$ where $\Delta^\star$ is the maximum degree of the optimal spanning tree.

Note that it is NP-Hard to find the value of $\Delta^\star$.
Our binary search does not find the exact value of $\Delta^\star$.
The final $k^\star$ might be equal to $\Delta^\star$ or $\Delta^\star-1$.
More precisely, although in the case of $k^\star = \Delta^\star - 1$, there is no spanning tree of maximum degree $k^\star$, and the analysis of the algorithm does not go through, it is possible that the algorithm successfully returns a spanning tree of maximum degree $k^\star+1 = \Delta^\star$.
Hence, the final spanning tree of maximum degree $k^\star+1$ might have maximum degree $\Delta^\star$ or $\Delta^\star+1$, and we can not distinguish these two cases.

\end{document}